\documentclass[12pt, draftclsnofoot, onecolumn]{IEEEtran}
\def\BibTeX{{\rm B\kern-.05em{\sc i\kern-.025em b}\kern-.08em
        T\kern-.1667em\lower.7ex\hbox{E}\kern-.125emX}}

\usepackage{amsmath, amsthm, amssymb}
\usepackage{bbding}
\usepackage[]{algorithm2e}
\usepackage[table]{xcolor}
\newtheorem{prop}{Proposition}
\usepackage{amsmath, amsthm, amssymb}
\usepackage{amsthm}
\usepackage{etoolbox}
\usepackage{amsmath,amsthm}

\usepackage{lipsum}
\usepackage{amsthm}
\usepackage{amsfonts}
\usepackage{booktabs}
\usepackage{graphicx}
\usepackage{epstopdf}
\usepackage{epsfig}
\usepackage{amssymb}
\usepackage{amsmath}
\usepackage{cite}
\usepackage{color,soul}
\usepackage{subfigure}
\usepackage{multirow}
\usepackage{rotating}
\usepackage{tabularx}
\usepackage{array}
\usepackage{dblfloatfix}
\usepackage{blindtext}
\usepackage[normalem]{ulem} 
\usepackage{soul} 

\usepackage{cleveref}
\usepackage{varioref}
\usepackage{amsmath,amsthm,amssymb}
\usepackage{bbding}
\usepackage[table]{xcolor}
\usepackage{etoolbox}
\AtEndEnvironment{theorem}{\null\hfill\qedsymbol}%
\usepackage{lipsum}
\usepackage{amsfonts}
\usepackage{booktabs}
\usepackage{epsfig}
\usepackage{cite}
\usepackage{color,soul}
\usepackage{subfigure}
\usepackage{multirow}
\usepackage{rotating}
\usepackage{tabularx}
\usepackage{array}
\usepackage{graphicx,dblfloatfix}
\usepackage{blindtext}
\usepackage[normalem]{ulem} 
\AtBeginDocument{%
}
\begin{document}
\title{ٔDynamic Frame Structure for Next Generation Wireless Networks}
\author{Mohammad~R.~Abedi, Mohammad~R.~Javan, Nader~Mokari, and Eduard~.~A.~Jorswieck
	\thanks{Manuscript received June 19, 2015.
		
		Mohammad~R.~Abedi, Nader~Mokari are with ECE Department, Tarbiat Modares University, Tehran, Iran.
		
		Mohammad~R.~Javan is with the Department of Electrical and Robotics Engineering, Shahrood University of Technology, Shahrood, Iran.
		
		E~.~A.~Jorswieck is with the Department of Systems and Computer Engineering, Dresden University of Technology (TUD), Germany. }}

\maketitle

\begin{abstract}
In this paper, we devise a novel radio resource block (RB) structure named dynamic resource block structure (D-RBS) which can handle low latency traffics and large fluctuations in data rates by exploiting smart time and frequency duplexing. In our framework, the main resource block with a predefined bandwidth and time duration is divided into several small blocks with the same bandwidth and time duration. Depending on the service requirements, e.g., data rate and latency, the users are assigned to some these small blocks which could be noncontiguous both in frequency and time. This is in contrast to the previously introduced static resource block structure (S-RBS) where the size of each RB is predetermined and fixed. We provide resource allocation frameworks for this RB structure and formulate the optimization problems whose solutions are obtained by alternate search method (ASM) based on successive convex approximation approach (SCA). We provide a global optimal solution by exploiting the monotonic optimization method. By simulation we study the performance of our proposed scheme with S-RBS scheme and show it has 26\% gain compared to the S-RBS scheme.
\end{abstract}


\section{Introduction}
\subsection{Motivations and State of the Art}
3GPP has categorized the network services into enhanced mobile broadband (eMBB),
 massive machine-type communications (mMTC), and
ultra-reliable and low-latency communications (URLLC) services \cite{3GPP}. eMBB services, like video streaming,
 require high data rate connections; mMTC services, like internet of things (IoT),
  require large number of devices to connect to the network which would only send small data payloads;
   and URLLC services, like tactile Internet (TI),  require communications of small payloads with low-latency and very high reliability.
    Due to heterogeneous requirements of these services, the efficient and dynamic resource allocation in the network is the main challenge. Traditional resource block structure (RBS) definition with fixed time-frequency structure, e.g., as in LTE with 1 millisecond transmission time
      interval (TTI) and 15 kHz bandwidth, is not able to support the network with these three types of services \cite{3GPP,3GPP-36-211}. The eMMB services require many frequency resources to satisfy their data rate requirements whose latency is not of importance, and hence, the time duration of the allocated RBs could be high, e.g., above one millisecond. On the other hand, URLLC services have small packets for transmission and require that these packets are sent in short time intervals, e.g., 0.1 to 0.5 milliseconds, with high reliability. In 4G, The admitted end-to-end latency of 4G networks is 30$\sim$100 milliseconds and the target reliability of information transmission is 0.99 which is not sufficient for ultra reliable services such as virtual reality and autonomous deriving, which is equal to 0.99999. In 5G, the RBs with different TTIs and bandwidths can satisfy latency and reliability requirements of services. In these networks, the notion of RB is defined as a time-frequency block with different time durations and bandwidths. In other words,
  the total time duration and bandwidth, e.g., $T$ milliseconds time duration and $W$ Hz bandwidth,
   is partitioned into several RBs with different time durations, e.g., $T$, $\frac{T}{2}$, and $\frac{T}{4}$, and different bandwidths,
    e.g., $W$, $\frac{W}{2}$, and $\frac{W}{4}$. Different RBs with different time durations and bandwidths can be used for
     different services based on the service requirements.

      In the aforementioned structure, although the size of the RBs is different, the size of each RB is predetermined and fixed. We call this structure by static RBS (S-RBS) as the size of the RBs is predefined and cannot be changed. However, this
RB structure, due to fixed TTI and bandwidth, does not
have enough flexibility to satisfy different service requirements
in an efficient way. Therefore, we propose a dynamic
       RBS (D-RBS) for 5G in which the total time duration and bandwidth is partitioned into several small RBs of
        equal time duration and bandwidth. Based on the service requirements, some small RBs are aggregated to form a composite RB\footnote{We call the aggregated small blocks which are allocated to a user by composite resource block} used by one specific service.
         In other words, the allocated RBs to each user can be constructed using one or more of these small RBs which could be spanned over different time and frequency bands. Dynamic here means that the size
          of the ultimate RBs is not predefined and could be different. The proposed structure enables dynamic allocation of RBs leading to
           efficient network resource utilization.
           
          Our proposed approach can also be easily extended to network slicing where each type of user belongs to one slice. Network slicing is a kind of virtual networking architecture. Its structure allows the creation of multiple virtual networks with different rate and latency requirements atop a shared physical infrastructure.


\subsection{Related Works}
The resource allocation problem among 5G network's services has been addressed in several works \cite{anand2017joint,pedersen2016flexible}.
 The authors in \cite{anand2017joint} study the downlink multiplexing of URLLC and eMBB services. The goal is to maximize the utility for eMBB traffic
  with respect to URLLC service requirement constraints. In \cite{pedersen2016flexible}, the authors study the downlink of 5G networks for URLLC services with the aim of achieving trade off between SE, latency, and reliability for each link and service flow. The dynamic scheduling for mMTC and URLLC services is studied in \cite{lien2017efficient} for
   NR where the authors investigate the performance of feedbackless
   and feedback based frameworks based on reinforcement learning. They illustrate that both schemes are able to
   effectively deploy both URLLC and mMTC services in NR. Recently, individual service designing has attracted a lot of attentions and there are several researches on RB design in
    wireless networks \cite{shafi20175g,dawy2017toward,popovski2018wireless,zhang2017network,anand2017joint}. An overview of 5G deployment challenges are provided in \cite{shafi20175g}. In \cite{dawy2017toward}, the authors provide a survey of key techniques to overcome the new requirements and challenges of mMTC applications. The principles of using diversity sources, design of packets, and access protocols to support URLLC are investigated in \cite{popovski2018wireless}. In \cite{zhang2017network}, the authors provide joint power and subchannel allocation for sliced 5G network with respect to both the inter-tier and intra-tier interference constraints. In \cite{mahmood2016radio}, the authors
     study orthogonal resource allocation for mMTC and eMBB. The uplink multiplexing of URLLC, eMBB, and mMTC services is studied in \cite{tian2017uplink}.
 In \cite{she2018cross}, the authors investigate downlink transmission design for URLLC services. Joint resource allocation in uplink and downlink based
  on effective bandwidth and effective capacity to ensure the quality of service (QoS) for URLLC is considered in \cite{aijaz2016towards,she2017radio}. Joint uplink and
   downlink bandwidth optimization with respect to delay constraints to guarantee both packet loss constraint and end-to-end delay requirement is considered in
    \cite{she2017radio}. In \cite{she2018joint}, the authors propose a packet delivery mechanism for URLLC. The goal is to reduce the bandwidth required to guaranteeing queuing delay based on statistical multiplexing queuing model. The network slicing based on orthogonal and non-orthogonal
      radio resource allocation for the three types of services of 5G is considered in \cite{popovski20185g}. In \cite{matthiesen2018throughput}, the authors optimize power and bandwidth
      allocation across radio access network slices and users which have heterogeneous QoS requirements. The goal is to maximize both throughput and energy
      efficiency in the sliced radio access network. In \cite{wu2017signal,lien2017efficient},
       the authors investigate the dynamic downlink
resource allocation for eMBB and URLLC services on the same
time/frequency resources. The impact of 5G frame structure on URLLC performance is investigated in
\cite{iwabuchi20175g}. The performance of flexible TTI to adopt
traffic load is
 investigated in \cite{liao2016resource,fountoulakis2017examination}. A 5G frame structure
designing is considered in \cite{pedersen2016flexible} to support user's service requirements. In all the previous works, the authors have considered
 a static structure for RBs that does not have enough flexibility to respond to different service requirements. Indeed, they do not consider both outage and bit error rate constraints that are critical for new emerging services in 5G. In fact, since the traditional framing, resource allocation, and user association schemes are not  flexible enough and do not consider various service types requirements in their optimization problems, they are not appropriate for 5G services.

\subsection{Main Contributions}
In this paper, we propose and design a dynamic framing scheme for supporting different types of
 services in the network, i.e., D-RBS, and investigate its performance. In contrast to the S-RBS scheme, in the proposed scheme,
  composite RBs are flexible in that the constructing small blocks could be of different time durations in different frequency bands, and the number of allocated small blocks us determined by the adopted resource allocation algorithm.

   In D-RBS, it is assumed that each cell partitions the total time duration and bandwidth in its own way which is determined by the resource allocation framework. The resulting partitioned resource blocks would be different for different cells. More precisely, we devise and develop a new resource management framework for wireless
     networks with different types of services each with different requirements where the time-frequency resource of the network is
     partitioned into several composite RBs each having different number of small RBs To this end, we formulate our proposed scheme as an optimization problem whose outcome is joint dynamic frame structure design, resource allocation, and user
       association (matching of users to RBs). Our aim is to maximize the network throughput under transmit power constraints and  QoS requirements
        of eMBB, mMTC, and URLLC services. Our contributions are as follows:
\begin{itemize}
    \item We design a new dynamic frame resource management scheme for next generations of wireless networks where different services have different QoS requirements which in turn requires different time-bandwidth resources. In the proposed scheme,  the resource management algorithm provides joint RB partitioning, resource allocation, and user association.
    \item We formulate our proposed resource management scheme as optimization problems which are mixed integer nonlinear nonconvex optimization problems in general. The alternate search method (ASM) is  used to decompose the main optimization problem into to optimization sub-problem. In the first optimization sub-problem, given RB assignment, we optimize the  power allocation. In the second optimization sub-problem, given the power allocation results of the first sub-problem, we optimize
    the RB association. To solve the resulting non-convex sub-problems at each level, we exploit the successive convex approximation (SCA) method to write the main non-convex problem into a series of convex problems which could be solved using standard tools like CVX.
    \item We study the convergence of our proposed scheme and show that the algorithms converge to a sub-optimal solution. We further investigate our proposed scheme from the computational complexity  prespective.
        \item We provide a global optimization solution by means of the monotonic optimization method.
    \item We study the performance of the proposed scheme using simulations for different network parameters. We show the superiority of our scheme and show that the performance of D-RBS is better than that of the S-RBS scheme due to its dynamic nature which is suitable for highly dynamic environment of wireless networks.  
\end{itemize}
The remainder of this paper is organized as follows. System model and descriptions regarding 5G services and requirements are presented in Section \ref{SystemModelandDescription}. Problem formulation and solution algorithms are provided in Section \ref{ProblemFormulationandSolution}. In Section \ref{monotonic}, we provide global optimization solution by monotonic optimization method. In Section \ref{ComputationalComplexity andConvergence Analysis}, we provide the convergence proof and the computational complexity of our scheme. Simulation results
are provided in Section \ref{SimulationResults}. And finally, Section \ref{conclusions} concludes this work.

\section{System Model and Description}\label{SystemModelandDescription}
\subsection{System Model}
We consider a multi-cell downlink of an OFDMA
network. There is one BS at each cell. The BS set is denoted by $\mathcal{B}=\{1,2,\dots,B\}$ with $|\mathcal{B}|=B$  where $|.|$ denotes the number of elements in a set. There are three types of users which request different types of services. In BS $b$, the sets of users which request eMBB, mMTC, and URLLC services are demoted by $\mathcal{K}^{\text{e}}_b=\{1,2,\dots,K^{\text{e}}_b\}$, $\mathcal{K}^{\text{m}}_b=\{1,2,\dots,K^{\text{m}}_b\}$ and $\mathcal{K}^{\text{u}}_b=\{1,2,\dots,K^{\text{u}}_b\}$, respectively, with $|\mathcal{K}^{\text{e}}_b|=K^{\text{e}}_b$, $|\mathcal{K}^{\text{m}}_b|=K^{\text{m}}_b$ and $|\mathcal{K}^{\text{u}}_b|=K^{\text{u}}_b$. The set of total user in BS $b$ is denoted by   $\mathcal{K}_b=\mathcal{K}^{\text{e}}_b\cup \mathcal{K}^{\text{m}}_b \cup \mathcal{K}^{\text{u}}_b$ and $K_b=|\mathcal{K}_b|$ denotes the total number of users in BS $b$. The BSs and users are equipped with one antenna. A time-frequency resource of $T$ seconds and $W$ Hz is used by each cell. In the proposed frame structure, we assume that the time-frequency resource is divided into several small RBs with the time duration  $\chi $ and frequency size $\vartheta$. All the small RBs could be shown by a matrix $[A]_{F\times N}$ with $N=\frac{T}{\chi }$ and $F=\frac{W}{\vartheta}$. In contrast to the S-RBS scheme, in which the size of RBs are constant (Fig. \ref{Frame-structure-2}), in the proposed scheme, at each cell, users are flexibly multiplexed over these small RBs. To each user several small RBs may be assigned which could be in different times and frequency bands and be noncontiguous in time and frequency. This RB structure is shown in Fig. \ref{Frame-structure-2}.


With carrier and slot aggregation capability, already in use \cite{3GPP-5G}, we can aggregate radio carriers (in the same band or across disparate bands) and slots of different small blocks to construct composite RBs to meet the requirements of users. Therefore, for user $k$ in cell $b$, the assigned composite RB is constructed by aggregation of one or more small blocks of time duration  $\chi $ and frequency size $\vartheta$. Then, the users are multiplexed in an orthogonal fashion to the composite RBs.
\begin{figure}[h]
    \begin{center}
        \includegraphics[width=6.5 in]{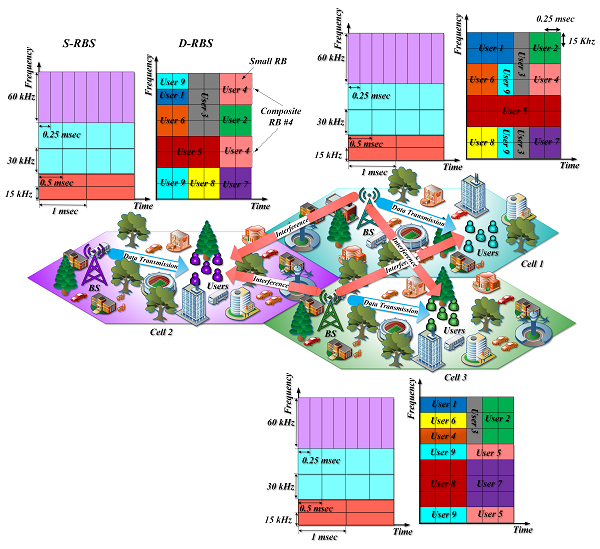} 
        \caption{The S-RBS scheme: there are nine RBs for each cell. The data transmission and interference vectors are illustrated by blue and red color. The D-RBS scheme: there are nine composite RBs for each cell. The data transmission and interference vectors are illustrated by blue and red color.}
        \label{Frame-structure-2}
    \end{center}
\end{figure}

\subsection{The eMBB Service}
The eMBB service requires high bandwidth and the reliability, given by packet error rate (PER), about $10^{-3}$ \cite{3GPP}. The instantaneous transmission rate between the $b^{\text{th}}$ BS and the $k^{\text{th}}$ user on RB $(f,n)$ is defined as
\begin{equation}
R^{fn}_{bk}(\textbf{p},\textbf{s},\textbf{g})=\chi \vartheta\log_2\left(1+\frac{p^{fn}_{bk}g^{fn}_{bk}}
{I^{fn}_{bk}(\textbf{p},\textbf{s},\textbf{g})+\vartheta N_0}\right),
\end{equation}
where $I^{fn}_{bk}(\textbf{p},\textbf{s},\textbf{g})=\sum_{i\in\mathcal{B}\setminus\{b\}}
\sum_{j\in\mathcal{K}_{i}\setminus\{k\}}s^{fn}_{ij}p^{fn}_{ij}g^{fn}_{ik}$ is the interference to user $k$ in cell $b$ on RB $(f,n)$, $p^{fn}_{bk}$ and $g^{fn}_{bk}$ are the transmit power and channel power gain of the $b^{\text{th}}$ BS to the $k^{\text{th}}$ user on RB $(f,n)$, respectively, while $N_0$ is the single-sided noise power-spectral-density (PSD). $\textbf{g}=[g^{11}_{11},\dots,g^{fn}_{bk},\dots,g^{FN}_{BK}]^T$, $\textbf{p}=\{p^{11}_{11},\dots,p^{fn}_{bk},\dots,p^{FN}_{BK}\}$, and  $\textbf{s}=\{s^{11}_{11},\dots,s^{fn}_{bk},\dots,s^{FN}_{BK}\}$ denote the channel gain, the power allocation and small RBs assignment vectors.

The binary-valued RB-association factor $s^{fn}_{bk}$ represents
both RB and BS assignment indicator for
user $k$ of BS $b$ on RB $(f,n)$, i.e., $s^{fn}_{bk}=1$ when BS $b$ allocates RB $(f,n)$ to user $k$, and $s^{fn}_{bk}=0$, otherwise. For user $k$ which requests eMBB service, the following constraint should be satisfied to ensure that the
data rate of user $k$ is equal or above the required minimum data rate:
\begin{equation}\label{eq-min-R-e}
\sum_{f\in\mathcal{F}}\sum_{n\in\mathcal{N}}s^{fn}_{bk}R^{fn}_{bk}(\textbf{p},\textbf{s},\textbf{g})\geq R^{\text{min,e}}_{bk}, \forall b\in\mathcal{B}, k\in\mathcal{K}^{\text{e}}_b,
\end{equation}
where $R^{\text{min,e}}_{bk}$ denotes the required data rate of eMBB user $k$ at BS $b$.

\subsection{The URLLC Service}
The URLLC service requires low-latency and very high reliability transmission with packet loss probability lower than $10^{-7}$ \cite{3GPP}. The achievable rate of user $k$ in RB $(f,n)$ with finite block length can be accurately approximated as follows \cite{durisi2016toward,yang2014quasi}:
\begin{equation}\label{eq-R_uRLLC}
R^{\text{URLLC},fn}_{bk}(\varepsilon^{fn}_{bk})\approx\frac{\chi \vartheta}{\ln 2} \left[\ln\left(1+\gamma^{fn}_{bk}(\textbf{p},\textbf{s},\textbf{g})\right)-\sqrt{\frac{\nu^{fn}_{bk}}{\vartheta}}f^{-1}_Q(\varepsilon^{fn}_{bk})\right] ,
\end{equation}
where $\gamma^{fn}_{bk}(\textbf{p},\textbf{s},\textbf{g})=\frac{p^{fn}_{bk}g^{fn}_{bk}}
{\sum_{i\in\mathcal{B}\setminus\{b\}}
    \sum_{j\in\mathcal{K}_{i}\setminus\{k\}}s^{fn}_{ij}p^{fn}_{ij}g^{fn}_{ik}+\vartheta N_0}$, $f^{-1}_Q(.)$  is the inverse of Gaussian-Q function, $\nu^{fn}_{bk}=1-\frac{1}{(1+\gamma^{fn}_{bk})^2}$ and $\varepsilon^{fn}_{bk}$ denotes the decoding error probability. The number of symbols in the block is $\psi=\chi  \vartheta$. When transmitting $\kappa$ bits from BS $b$ to user $k$ in the short blocklength regime, by setting $\chi  R^{\text{URLLC},fn}_{bk}=\kappa$, the decoding error probability can be obtained from (\ref{eq-R_uRLLC}) as follows:
\begin{equation}\label{eq-R_uRLLC-2}
\varepsilon^{fn}_{bk}\approx \mathbb{E}_{\textbf{g}}\left\lbrace f_Q\left(\sqrt{\frac{\vartheta}{\nu^{fn}_{bk}}} \left[ \ln\left(1+\gamma^{fn}_{bk}(\textbf{p},\textbf{s},\textbf{g}) \right)-\frac{\kappa \ln 2}{\chi \vartheta} \right] \right) \right\rbrace.
\end{equation}
where $\mathbb{E}(x)$ denotes the expected value of $x$. Unfortunately, a closed-form expression for Q-function does not exist. Hence, we utilize an approximation of $f_Q\left( \frac{\log_2\left( 1+\gamma^{fn}_{bk}(\textbf{p},\textbf{s},\textbf{g})\right) -\hat{\varphi}}{\sqrt{\nu^{fn}_{bk}(\gamma^{fn}_{bk}(\textbf{p},\textbf{s},\textbf{g}))(\log_2 e)^2/\psi}}\right) \approx \Gamma(\gamma^{fn}_{bk}(\textbf{p},\textbf{s},\textbf{g}))$ to approximate the decoding error probability \cite{makki2014finite}, where $\hat{\varphi}=\kappa/\psi$ is the number of bits in each symbol,
\begin{equation}
\Gamma(\gamma^{fn}_{bk}(\textbf{p},\textbf{s},\textbf{g}))=
\begin{cases}
1 & \quad \gamma^{fn}_{bk}\leq \eta_1,\\
1/2-\varpi\sqrt{\psi}(\gamma^{fn}_{bk}(\textbf{p},\textbf{s},\textbf{g})-\theta) &\quad \eta_1 \leq \gamma^{fn}_{bk} \leq \eta_2,\\
0 &\quad \gamma^{fn}_{bk}\geq \eta_2,
\end{cases}
\end{equation}
where $\varpi=\frac{1}{2\pi\sqrt{2^{2\hat{\varphi}}-1}}$, $\theta=2^{\hat{\varphi}}-1$, $\eta_1=\theta-\frac{1}{2\varpi \sqrt{\psi}}$, and $\eta_2=\theta+\frac{1}{2\varpi \sqrt{\psi}}$. 

The following constraint should be satisfied to ensure that the packet loss probability of user $k$ in RB $(f,n)$ is equal or below the threshold value $\varepsilon^{\text{max},fn}_{bk}$:
\begin{equation}\label{eq-packet-loss}
\mathbb{E}_{\textbf{g}}\left\lbrace \Gamma(\gamma^{fn}_{bk}(\textbf{p},\textbf{s},\textbf{g})\right\rbrace\leq \varepsilon^{\text{max},fn}_{bk}, \forall f\in\mathcal{F}, n\in\mathcal{N}, b\in\mathcal{B}, k\in\mathcal{K}^{\text{u}}_b.
\end{equation}

Due to the delay limitation of URLLC users, the following constraint is used:
\begin{equation}\label{eq-r21}
\left( \sum_{f\in\mathcal{F}}s^{fn}_{bk}\right) \left(\sum_{f\in\mathcal{F}}\sum_{\acute{n}\in\mathcal{N},\acute{n}\neq n}s^{f\acute{n}}_{bk} \right) = 0,\forall n\in\mathcal{N}, k\in\mathcal{K}^{\text{u}}_b, b\in\mathcal{B}.
\end{equation}
which means that each URLLC user is assigned to only one RBs column (i.e., one column of small RBs over the entire spectrum and one TTI). More specifically, (\ref{eq-r21}) ensures when any RBs in one RB column is assigned to one URLLC user, that URLLC user would not be assigned any RBs in other RB columns.

\subsection{The mMTC Service}
The mMTC users require fixed, typically low, transmission rate and PER on the order of $10^{-1}$ \cite{3GPP}.  To guarantee
these requirements, the following constraints should be applied:
\begin{equation}\label{eq-min-R-m}
\sum_{f\in\mathcal{F}}\sum_{n\in\mathcal{N}}s^{fn}_{bk}R^{fn}_{bk}(\textbf{p},\textbf{s},\textbf{g})\geq R^{\text{min,m}}_{bk}, \forall b\in\mathcal{B}, k\in\mathcal{K}^{\text{m}}_b,
\end{equation}
and the following constraint should be satisfied to ensure that the packet loss probability of user $k$ in RB $(f,n)$ is equal or below the threshold values $\tilde{\varepsilon}^{\text{max},fn}_{bk}$:
\begin{equation}\label{eq-packet-loss-2}
\mathbb{E}_{\textbf{g}}\left\lbrace \Gamma(\gamma^{fn}_{bk}(\textbf{p},\textbf{s},\textbf{g})\right\rbrace\leq \tilde{\varepsilon}^{\text{max},fn}_{bk}, \forall f\in\mathcal{F}, n\in\mathcal{N}, b\in\mathcal{B}, k\in\mathcal{K}^{\text{m}}_b.
\end{equation}

\section{Problem Formulation and Solution}\label{ProblemFormulationandSolution}
We aim to maximize the total network data rate of D-RBS scheme which can be formulated as follows:

\emph{\textbf{Problem}} $\mathcal{P}^{\text{D-RBS}}$
\begin{subequations}\label{eq-main}
    \begin{align}
    \max_{\textbf{p},\textbf{s}}~~&\sum_{f\in\mathcal{F}}\sum_{n\in\mathcal{N}}\sum_{b\in\mathcal{B}}\sum_{k\in\mathcal{K}}s^{fn}_{bk}R^{fn}_{bk}(\textbf{p},\textbf{s},\textbf{g}),
    \\\label{eqq-1}\text{s.t.}~~&\sum_{f\in\mathcal{F}}\sum_{n\in\mathcal{N}}\sum_{k\in\mathcal{K}_b}s^{fn}_{bk}p^{fn}_{bk}\leq P^{\text{max}}_b, \forall b\in\mathcal{B},\\\label{eqq-2}&
    \sum_{k\in\mathcal{K}_b}s^{fn}_{bk}\leq 1,\forall f\in\mathcal{F}, n\in\mathcal{N}, b\in\mathcal{B},
    \\\label{eqq-3}& p^{fn}_{bk}\geq 0, s^{fn}_{bk}\in\{0,1\}, \forall f\in\mathcal{F}, n\in\mathcal{N}, b\in\mathcal{B}, k\in\mathcal{K},\\\nonumber&
    (\ref{eq-min-R-e}), (\ref{eq-packet-loss})
    -(\ref{eq-packet-loss-2}).
    \end{align}
\end{subequations}
(\ref{eqq-1}) and (\ref{eqq-2}) denote the transmit power constraint of each BS and the OFDMA exclusive small RB allocation
in each BS $b$, respectively\footnote{Note the dependency of the problem on the channel realization. This means that the problem should be solved for each channel realization.}. Due to the delay limitation of URLLC users, (\ref{eq-r21}) is used to restrict the access of each URLLC user.

Note that the optimization problem (\ref{eq-main}) is non-convex mixed-integer which makes it hard to develop an efficient algorithm to solve it globally. Note the dependency of the problem on the channel realization. This means that the problem should be solved for each channel realization.

By exploiting ASM, we propose a two-step iterative algorithm to optimize power allocation and small RB assignment in each BS. In the first step, given RB assignment variables, the optimization problem is solved to find power allocation vector. In the second step, given power allocation results, we obtain the RB assignment. The iteration of algorithm can be stopped when the difference of the objective function values in two consecutive iterations is small enough.

However, both the power allocation and small RB assignment problems are non-convex. We use complementary geometric programming (CGP) \cite{boyd2004convex,shen2004global} to solve the corresponding optimization problem. The proposed algorithm to obtain the power allocation and small RBs assignment is shown in Table \ref{Iterative Optimization Algorithm}.

\begin{table}[h!]\caption{An Iterative Algorithm in Two Steps to Obtain the Power Allocation and small RBs Assignment}\label{Iterative Optimization Algorithm}
    \centering
    \begin{tabular}{p{15cm}}
        \toprule
        \textbf{Algorithm I}: Iterative Joint Power Allocation and small RBs Assignment Algorithm\\
        \midrule
        \textbf{Initialization}: Select a starting point $\textbf{s}(0)$, and set iteration number
        $\varrho=0$;\\
        \textbf{Repeat}\\~~~
        \textbf{Step1: Power Allocation}\\~~~
        \textbf{Initialization}: Set  $\textbf{s}(\varrho_1)=\textbf{s}(\varrho)$, $\textbf{p}(\varrho_1)=\textbf{p}(\varrho)$ and $\varrho_1=0$;\\~~~
        \textbf{Repeat}\\~~~~~~
        \textbf{Step1.1}: Update $\mu^{fn}_{bk}(\varrho_1)$, and $\mu_0(\varrho_1)$ using (\ref{eq_mu_f}) and (\ref{eq_mu_0});\\~~~~~~
        \textbf{Step1.2}: Solve (\ref{eq_power}) to find optimal power allocation  $\mathbf{p}(\varrho_1)$ (Convex programming via CVX);\\~~~~~~
        \textbf{Step1.3}: $\varrho_1=\varrho_1+1$;\\~~~
        \textbf{Until}: $\|\sum_{f\in\mathcal{F}}\sum_{n\in\mathcal{N}}\sum_{b\in\mathcal{B}}\sum_{k\in\mathcal{K}}s^{fn}_{bk}R^{fn}_{bk}(\textbf{p}(\varrho_1),\textbf{s}(\varrho),\textbf{g})-\sum_{f\in\mathcal{F}}\sum_{n\in\mathcal{N}}\sum_{b\in\mathcal{B}}\sum_{k\in\mathcal{K}}s^{fn}_{bk}R^{fn}_{bk}(\textbf{p}(\varrho_1-1),\textbf{s}(\varrho),\textbf{g})\|\leq {\epsilon}_1$;\\~~~
        \textbf{Return}: $\textbf{p}(\varrho)=\textbf{p}(\varrho_1)$\\~~~
        \textbf{Step2: RB Assignment}\\~~~
        \textbf{Initialization}: Set  $\textbf{p}(\varrho_2)=\textbf{p}(\varrho)$, and $\varrho_2=0$;\\~~~
        \textbf{Repeat}\\~~~~~~
        \textbf{Step2.1}: Update $\phi^n_{bk}(\varrho_2)$, $\xi^n_{bk}(\varrho_2)$, $\nu^{fn}_{bk}(\varrho_2)$, $\delta^{fn}_{bk}(\varrho_2)$, $\varphi^{fn}_{bk}(\varrho_2)$, $d_0(\varrho_2)$, and $d^{fn}_{bk}(\varrho_2)$, using (\ref{eq_phi}) and (\ref{eq-varphi})-(\ref{eq-d_f});\\~~~~~~
        \textbf{Step2.2}: For the obtained $\mathbf{p}(\varrho)$,
        solve (\ref{eq-main-RBA}) to find $\textbf{s}(\varrho+1)$ (Convex programming via CVX);\\~~~~~~
        \textbf{Step2.3}: $\varrho_2=\varrho_2+1$;\\~~~
        \textbf{Until} $\|\sum_{f\in\mathcal{F}}\sum_{n\in\mathcal{N}}\sum_{b\in\mathcal{B}}\sum_{k\in\mathcal{K}}s^{fn}_{bk}(\varrho_2)R^{fn}_{bk}(\textbf{p}(\varrho),\textbf{s}(\varrho_2),\textbf{g})$\\~~~~~~~~~~~$-\sum_{f\in\mathcal{F}}\sum_{n\in\mathcal{N}}\sum_{b\in\mathcal{B}}\sum_{k\in\mathcal{K}}s^{fn}_{bk}(\varrho_2-1)R^{fn}_{bk}(\textbf{p}(\varrho),\textbf{s}(\varrho_2-1),\textbf{g})\|\leq {\epsilon}_2$;\\~~~
        \textbf{Return}: $\textbf{s}(\varrho)=\textbf{s}(\varrho_2)$\\~~~
        \textbf{Step3:} $\varrho=\varrho+1$;\\
        \textbf{Until}: $\|\sum_{f\in\mathcal{F}}\sum_{n\in\mathcal{N}}\sum_{b\in\mathcal{B}}\sum_{k\in\mathcal{K}}s^{fn}_{bk}R^{fn}_{bk}(\textbf{p}(\varrho),\textbf{s}(\varrho),\textbf{g})-\sum_{f\in\mathcal{F}}\sum_{n\in\mathcal{N}}\sum_{b\in\mathcal{B}}\sum_{k\in\mathcal{K}}s^{fn}_{bk}R^{fn}_{bk}(\textbf{p}(\varrho-1),\textbf{s}(\varrho),\textbf{g})\|\leq {\epsilon}_1$ and\\~~~~~~~ $\|\sum_{f\in\mathcal{F}}\sum_{n\in\mathcal{N}}\sum_{b\in\mathcal{B}}\sum_{k\in\mathcal{K}}s^{fn}_{bk}(\varrho)R^{fn}_{bk}(\textbf{p}(\varrho),\textbf{s}(\varrho),\textbf{g})-\sum_{f\in\mathcal{F}}\sum_{n\in\mathcal{N}}\sum_{b\in\mathcal{B}}\sum_{k\in\mathcal{K}}s^{fn}_{bk}(\varrho-1)R^{fn}_{bk}(\textbf{p}(\varrho),\textbf{s}(\varrho-1),\textbf{g})\|\leq {\epsilon}_2$;\\
        \textbf{Return}:
        $(\mathbf{p}^*,\textbf{s}^*)=(\mathbf{p}(\varrho),\textbf{s}(\varrho))$.\\
        \bottomrule
        \label{Alternate Optimization Algorithm}
    \end{tabular}
    \vspace{-0.5cm}
\end{table}

\subsection{Power Allocation Sub-Problem}
Given $\textbf{s}$, the following power allocation optimization sub-problem to maximize the total rate should be solved:

\emph{\textbf{Problem}} $\mathcal{P}^{\text{D-RBS}}_{\text{PA}}$
\begin{subequations}\label{eq-main-PA}
    \begin{align}
    \max_{\textbf{p}(\varrho_1)}~~&\sum_{f\in\mathcal{F}}\sum_{n\in\mathcal{N}}\sum_{b\in\mathcal{B}}\sum_{k\in\mathcal{K}}s^{fn}_{bk}(\varrho)R^{fn}_{bk}(\textbf{p}(\varrho_1),\textbf{s}(\varrho),\textbf{g}),
    \\\label{eq-main-PA-1}\text{s.t.}~~&\sum_{f\in\mathcal{F}}\sum_{n\in\mathcal{N}}\sum_{k\in\mathcal{K}_b}s^{fn}_{bk}(\varrho)p^{fn}_{bk}(\varrho_1)\leq P^{\text{max}}_b, \forall b\in\mathcal{B},\\\label{eq-main-PA-6}& p^{fn}_{bk}(\varrho_1)\geq 0, \forall f\in\mathcal{F}, n\in\mathcal{N}, b\in\mathcal{B}, k\in\mathcal{K}.
    \\\nonumber&
    (\ref{eq-min-R-e}), (\ref{eq-packet-loss}), (\ref{eq-min-R-m}),(\ref{eq-packet-loss-2}).
    \end{align}
\end{subequations}

Optimization problem (\ref{eq-main-PA}) is non-convex, due to the interference term in the objective function and constraints (\ref{eq-min-R-e}), (\ref{eq-packet-loss}), (\ref{eq-min-R-m}), and (\ref{eq-packet-loss-2}). We convert (\ref{eq-main-PA}) into the GP
optimization problem. In this regards, we rewrite the objective of (\ref{eq-main-PA}) as:
\begin{align}
\max_{\textbf{p}(\varrho_1)}~~&\prod_{f\in\mathcal{F},n\in\mathcal{N},b\in\mathcal{B},k\in\mathcal{K}}\Gamma^{fn}_{bk}(\textbf{p}(\varrho_1),\textbf{s}(\varrho),\textbf{g}),
\end{align}
where $\Gamma^{fn}_{bk}(\textbf{p}(\varrho_1),\textbf{s}(\varrho),\textbf{g})=\frac{\vartheta N_0+I^{fn}_{bk}(\textbf{p}(\varrho_1),\textbf{s}(\varrho),\textbf{g})+p^{fn}_{bk}(\varrho_1)g^{fn}_{bk}}{\vartheta N_0+I^{fn}_{bk}(\textbf{p}(\varrho_1),\textbf{s}(\varrho),\textbf{g})}$. Arithmetic-geometric mean approximation (AGMA) method is exploited for posynomial form approximation \cite{shen2004global}. Therefore, $[\Gamma^{fn}_{bk}(\textbf{p}(\varrho_1),\textbf{s}(\varrho),\textbf{g})]^{-1}$ can be
approximated as \cite{shen2004global}
\begin{align}
&\hat{\Gamma}^{fn}_{bk}(\textbf{p}(\varrho_1),\textbf{s}(\varrho),\textbf{g})\\\nonumber&=\left(\vartheta N_0+I^{fn}_{bk}(\textbf{p}(\varrho_1),\textbf{s}(\varrho),\textbf{g})\right)\left(\frac{\vartheta N_0}{\mu_0(\varrho_1)}\right)^{-\mu_0(\varrho_1)}\prod_{f\in\mathcal{F},n\in\mathcal{N},b\in\mathcal{B},k\in\mathcal{K}}\left(\frac{p^{fn}_{bk}(\varrho_1)g^{fn}_{bk}}{\mu^{fn}_{bk}(\varrho_1)}\right)^{-\mu^{fn}_{bk}(\varrho_1)},
\end{align}
where
\begin{equation}\label{eq_mu_f}
\mu^{fn}_{bk}(\varrho_1)=\frac{p^{fn}_{bk}(\varrho_1-1)g^{fn}_{bk}}{\vartheta N_0+\sum_{f\in\mathcal{F}}
    \sum_{n\in\mathcal{N}}\sum_{b\in\mathcal{B}}
    \sum_{k\in\mathcal{K}_{b}}p^{fn}_{bk}(\varrho_1-1)g^{fn}_{bk}},
\end{equation}
\begin{equation}\label{eq_mu_0}
\mu_0(\varrho_1)=\frac{\vartheta N_0}{\vartheta N_0+\sum_{f\in\mathcal{F}}
    \sum_{n\in\mathcal{N}}\sum_{b\in\mathcal{B}}
    \sum_{k\in\mathcal{K}_{b}}p^{fn}_{bk}(\varrho_1-1)g^{fn}_{bk}}.
\end{equation}

Furthermore, we provide convex approximations for non-convex constraints (\ref{eq-packet-loss}) and (\ref{eq-packet-loss-2}). To this end, we use the first order Taylor series approximation of $\gamma^{fn}_{bk}(\textbf{p}(\varrho_1),\textbf{s}(\varrho),\textbf{g})$ around $\textbf{p}(\varrho_1-1)$, i.e., $\tilde{\gamma}^{fn}_{bk}(\textbf{p}(\varrho_1))$, as follows:
\begin{align}\label{eq-00121}
\tilde{\gamma}^{fn}_{bk}(\textbf{p}(\varrho_1),\textbf{s}(\varrho),\textbf{g}) \approx \gamma^{fn}_{bk}(\textbf{p}(\varrho_1),\textbf{s}(\varrho),\textbf{g})+\nabla\gamma^{fn}_{bk}(\textbf{p}(\varrho_1-1),\textbf{s}(\varrho),\textbf{g})\big(\textbf{p}(\varrho_1)-\textbf{p}(\varrho_1-1)\big),
\end{align}
where the gradient $\nabla \gamma^{fn}_{bk}(\textbf{p}(\varrho_1),\textbf{s}(\varrho),\textbf{g})$
with respect to $\mathbf{p}$ is given by
\begin{align}\label{eq-53}
\nabla \gamma^{fn}_{bk}(\textbf{p}(\varrho_1-1))=\left[\frac{\partial
	\gamma^{fn}_{bk}(\textbf{p}(\varrho_1-1))}{\partial p^{11}_{11}},\dots,\frac{\partial
	\gamma^{fn}_{bk}(\textbf{p}(\varrho_1-1))}{\partial p^{fn}_{bk}},\dots,\frac{\partial
	\gamma^{fn}_{bk}(\textbf{p}(\varrho_1-1))}{\partial p^{FN}_{BK}}\right] ,
\end{align}
where
\begin{equation}
\frac{\partial
	\gamma^{fn}_{bk}(\textbf{p}(\varrho_1-1),\textbf{s}(\varrho),\textbf{g})}{\partial p^{fn}_{bk}}=\frac{1}{\ln2}
\frac{g^{fn}_{bk}}{\vartheta N_0+\sum_{f\in\mathcal{F}}
	\sum_{n\in\mathcal{N}}\sum_{b\in\mathcal{B}}
	\sum_{k\in\mathcal{K}_{b}}p^{fn}_{bk}(\varrho_1-1)
	g^{fn}_{bk}}.
\end{equation}
Thus, the constraints (\ref{eq-packet-loss}) and (\ref{eq-packet-loss-2}) can be rewritten as follows:
\begin{equation}\label{eq-0012}
\mathbb{E}_{\textbf{g}}\left\lbrace \Gamma(\tilde{\gamma}^{fn}_{bk}(\textbf{p}(\varrho_1),\textbf{s}(\varrho),\textbf{g}))\right\rbrace\leq \varepsilon^{\text{max},fn}_{bk}, \forall f\in\mathcal{F}, n\in\mathcal{N}, b\in\mathcal{B}, k\in\mathcal{K}^{\text{u}}_b,
\end{equation}
\begin{equation}\label{eq-0013}
\mathbb{E}_{\textbf{g}}\left\lbrace \Gamma(\tilde{\gamma}^{fn}_{bk}(\textbf{p}(\varrho_1),\textbf{s}(\varrho),\textbf{g}))\right\rbrace\leq \tilde{\varepsilon}^{\text{max},fn}_{bk}, \forall f\in\mathcal{F}, n\in\mathcal{N}, b\in\mathcal{B}, k\in\mathcal{K}^{\text{m}}_b.
\end{equation}

Consequently, (\ref{eq-main-PA}) can be transformed into standard GP problem as follows:

\emph{\textbf{Problem}} $\check{\mathcal{P}}^{\text{D-RBS}}_{\text{PA}}$
\begin{subequations}\label{eq_power}
\begin{align}
\max_{\textbf{p}(\varrho_1)}~~&\prod_{f\in\mathcal{F},n\in\mathcal{N},b\in\mathcal{B},k\in\mathcal{K}}\hat{\Gamma}^{fn}_{bk}(\textbf{p}(\varrho_1),\textbf{s}(\varrho),\textbf{g})\\\text{s.t.}~~&\prod_{f\in\mathcal{F},n\in\mathcal{N}}\hat{\Gamma}^{fn}_{bk}(\textbf{p}(\varrho_1),\textbf{s}(\varrho),\textbf{g})\leq 2^{-R^{\text{min,e}}_{bk}}, \forall b\in\mathcal{B}, k\in\mathcal{K}^{\text{e}}_b,\\&\prod_{f\in\mathcal{F},n\in\mathcal{N}}\hat{\Gamma}^{fn}_{bk}(\textbf{p}(\varrho_1),\textbf{s}(\varrho),\textbf{g})\leq 2^{-R^{\text{min,m}}_{bk}}, \forall b\in\mathcal{B}, k\in\mathcal{K}^{\text{m}}_b,\\&\sum_{f\in\mathcal{F}}\sum_{n\in\mathcal{N}}\sum_{k\in\mathcal{K}_b}p^{fn}_{bk}(\varrho_1)\leq P^{\text{max}}_b, \forall b\in\mathcal{B},\\~~&\nonumber (\ref{eq-0012}), (\ref{eq-0013}).
\end{align}
\end{subequations}

\subsection{Resource Block Assignment}
With the value of $\textbf{p}(\varrho)$ from power allocation problem, the following RB assignment optimization problem is solved:

\emph{\textbf{Problem}} $\mathcal{P}^{\text{D-RBS}}_{\text{RBA}}$
\begin{subequations}\label{eq-main-RBA-22}
    \begin{align}\label{eq-main-RBA-0}
    \max_{\textbf{s}(\varrho_2)}~~&\sum_{f\in\mathcal{F}}\sum_{n\in\mathcal{N}}\sum_{b\in\mathcal{B}}\sum_{k\in\mathcal{K}}s^{fn}_{bk}(\varrho_2)R^{fn}_{bk}(\textbf{p}(\varrho),\textbf{s}(\varrho_2),\textbf{g}),
    \\\label{eq-main-RBA-1}\text{s.t.}~~&\sum_{f\in\mathcal{F}}\sum_{n\in\mathcal{N}}\sum_{k\in\mathcal{K}_b}s^{fn}_{bk}(\varrho_2)p^{fn}_{bk}(\varrho)\leq P^{\text{max}}_b, \forall b\in\mathcal{B},\\&s^{fn}_{bk}(\varrho_2)\in\{0,1\}, \forall f\in\mathcal{F}, n\in\mathcal{N}, b\in\mathcal{B}, k\in\mathcal{K},\\\nonumber&
    (\ref{eq-min-R-e}), (\ref{eq-packet-loss}),(\ref{eq-r21}), (\ref{eq-min-R-m}),(\ref{eq-packet-loss-2}).
    \end{align}
\end{subequations}

Due to objective function and constraints (\ref{eq-main-RBA-0}), (\ref{eq-min-R-e}), and (\ref{eq-min-R-m}), the optimization problem is non-convex. We first relax discrete variable $s^{fn}_{bk}(\varrho_2)$ into continuous one as $s^{fn}_{bk}(\varrho_2)\in[0,1]$. By exploiting AGMA, we transform problem (\ref{eq-main-RBA-22}) to a standard for of GP. Note that in the same way in (\ref{eq-00121}), we can use the first order Taylor series approximation of $\gamma^{fn}_{bk}(\textbf{p}(\varrho),\textbf{s}(\varrho_2),\textbf{g})$ around $s^{fn}_{bk}(\varrho_2-1)$. Defining $\alpha^n_{bk}(\varrho_2)=\sum_{f\in\mathcal{F}}s^{fn}_{bk}(\varrho_2)$ and $\beta_{bk}(\varrho_2)=\sum_{f\in\mathcal{F}}\sum_{n\in\mathcal{N}}s^{fn}_{bk}(\varrho_2)$, one can approximate (\ref{eq-r21}) by the following constraints:
\begin{subequations}
    \begin{align}\label{eq-uu-1}
    &(\omega^{n}_{bk}(\varrho_2))^{-1}+\alpha^n_{bk}(\varrho_2)\beta_{bk}(\varrho_2)(\omega^{n}_{bk}(\varrho_2))^{-1}\leq 1,\forall n\in\mathcal{N}, b\in\mathcal{B}, k\in\mathcal{K}^{\text{u}}_b,\\\label{eq-uu-2}&
    \left[\frac{1}{\phi^n_{bk}(\varrho_2)}\right]^{-\phi^n_{bk}(\varrho_2)}\omega^{n}_{bk}(\varrho_2)\left[\frac{(\alpha^n_{bk}(\varrho_2))^2}{\xi^n_{bk}(\varrho_2)}\right]^{-\xi^n_{bk}(\varrho_2)}\leq 1,\forall n\in\mathcal{N}, b\in\mathcal{B}, k\in\mathcal{K}^{\text{u}}_b,\\\label{eq-uu-3}&
    \alpha^n_{bk}(\varrho_2)\prod_{f\in\mathcal{F}}\left[\frac{s^{fn}_{bk}(\varrho_2)}{\nu^n_{bk}(\varrho_2)}\right]^{-\nu^n_{bk}(\varrho_2)} =1, \forall n\in\mathcal{N}, b\in\mathcal{B}, k\in\mathcal{K}^{\text{u}}_b,\\\label{eq-uu-4}&
    \beta_{bk}(\varrho_2)\prod_{f\in\mathcal{F},n\in\mathcal{N}}\left[\frac{s^{fn}_{bk}(\varrho_2)}{\delta^n_{bk}(\varrho_2)}\right]^{-\delta^n_{bk}(\varrho_2)} =1, \forall n\in\mathcal{N}, b\in\mathcal{B}, k\in\mathcal{K}^{\text{u}}_b,
    \end{align}
\end{subequations}
where $\omega^{n}_{bk}(\varrho_2)$ is an an auxiliary variable,
\begin{align}\label{eq_phi}
&\phi^n_{bk}(\varrho_2)=\frac{1}{(\alpha^n_{bk}(\varrho_2-1))^2+1},~
\xi^n_{bk}(\varrho_2)=\frac{(\alpha^n_{bk}(\varrho_2-1))^2}{(\alpha^n_{bk}(\varrho_2-1))^2+1},~\\&
\nu^{fn}_{bk}(\varrho_2)=\frac{s^{fn}_{bk}(\varrho_2-1)}{\alpha^n_{bk}(\varrho_2-1)},~
\delta^{fn}_{bk}(\varrho_2)=\frac{s^{fn}_{bk}(\varrho_2-1)}{\beta_{bk}(\varrho_2-1)}.
\end{align}

Based on (\ref{eq-uu-1})-(\ref{eq-uu-4}), (\ref{eq-r21}) is replaced by monomial equalities and posynomial inequalities. Next, we convert the objective function and write it into monomial form. By defining the auxiliary variables $\zeta_1$ and $\zeta_2$, (\ref{eq-main-RBA-22}) is written into the following standard for of GP:

\emph{\textbf{Problem}} $\check{\mathcal{P}}^{\text{D-RBS}}_{\text{RBA}}$
\begin{subequations}\label{eq-main-RBA}
    \begin{align}\label{eq-main-RBA-0-1}
    &\max_{\textbf{s}(\varrho_2),\zeta_1(\varrho_2),\boldsymbol{\omega}(\varrho_2),\boldsymbol{\alpha}(\varrho_2),\boldsymbol{\beta}(\varrho_2)}~~\zeta_1,
    \\\label{eq-main-RBA-1-1}\text{s.t.}~~&\zeta_2\left[\frac{\zeta_1(\varrho_2)}{d_0(\varrho_2)}\right]^{-d_0(\varrho_2)}\prod_{f\in\mathcal{F},n\in\mathcal{N},b\in\mathcal{B},k\in\mathcal{K}}\left[\frac{s^{fn}_{bk}(\varrho_2)R^{fn}_{bk}(\textbf{p}(\varrho),\textbf{s}(\varrho_2),\textbf{g})}{d^{fn}_{bk}(\varrho_2)}\right]^{-d^{fn}_{bk}(\varrho_2)}\leq 1,\\\label{eq-main-RBA-4}&
    R^{\text{min,m}}_{bk}\times \prod_{f\in\mathcal{F},n\in\mathcal{N}}\left[\frac{s^{fn}_{bk}(\varrho_2)R^{fn}_{bk}(\textbf{p}(\varrho),\textbf{s}(\varrho_2),\textbf{g})}{\varphi^{fn}_{bk}(\varrho_2)}\right]^{-\varphi^{fn}_{bk}(\varrho_2)}\leq 1, \forall b\in\mathcal{B}, k\in\mathcal{K}^{\text{m}}_b,\\\label{eq-main-RBA-5}&
    R^{\text{min,e}}_{bk}\times \prod_{f\in\mathcal{F},n\in\mathcal{N}}\left[\frac{s^{fn}_{bk}(\varrho_2)R^{fn}_{bk}(\textbf{p}(\varrho),\textbf{s}(\varrho_2),\textbf{g})}{\varphi^{fn}_{bk}(\varrho_2)}\right]^{-\varphi^{fn}_{bk}(\varrho_2)}\leq 1, \forall b\in\mathcal{B}, k\in\mathcal{K}^{\text{e}}_b,\\\label{eq-main-RBA-6}&
    \sum_{k\in\mathcal{K}_b}s^{fn}_{bk}(\varrho_2)\leq 1,\forall f\in\mathcal{F}, n\in\mathcal{N}, b\in\mathcal{B},\\\label{eq-main-RBA-7}&s^{fn}_{bk}(\varrho_2)\in[0,1], \forall f\in\mathcal{F}, n\in\mathcal{N}, b\in\mathcal{B}, k\in\mathcal{K},\\\nonumber
    &(\ref{eq-0012}),(\ref{eq-0013}), (\ref{eq-uu-1})-(\ref{eq-uu-4}).
    \end{align}
\end{subequations}
where $\zeta_2$ is a sufficiently large constant and
\begin{align}\label{eq-varphi}
&\varphi^{fn}_{bk}(\varrho_2)=\frac{s^{fn}_{bk}(\varrho_2-1)R^{fn}_{bk}(\textbf{p}(\varrho))}{\sum_{f\in\mathcal{F}}\sum_{n\in\mathcal{N}}\sum_{b\in\mathcal{B}}\sum_{k\in\mathcal{K}}s^{fn}_{bk}(\varrho_2-1)R^{fn}_{bk}(\textbf{p}(\varrho),\textbf{s}(\varrho_2),\textbf{g})},\\\label{eq-d_0}
&d_0(\varrho_2)=\frac{\zeta_1(\varrho_2-1)}{\zeta_1(\varrho_2-1)+\sum_{f\in\mathcal{F}}\sum_{n\in\mathcal{N}}\sum_{b\in\mathcal{B}}\sum_{k\in\mathcal{K}}s^{fn}_{bk}(\varrho_2-1)R^{fn}_{bk}(\textbf{p}(\varrho),\textbf{s}(\varrho_2),\textbf{g})},\\\label{eq-d_f}
&d^{fn}_{bk}(\varrho_2)=\frac{s^{fn}_{bk}(\varrho_2-1)R^{fn}_{bk}(\textbf{p}((\varrho)))}{\zeta_1(\varrho_2-1)+\sum_{f\in\mathcal{F}}\sum_{n\in\mathcal{N}}\sum_{b\in\mathcal{B}}\sum_{k\in\mathcal{K}}s^{fn}_{bk}(\varrho_2-1)R^{fn}_{bk}(\textbf{p}(\varrho),\textbf{s}(\varrho_2),\textbf{g})}.
\end{align}

Based on Algorithm 1, The optimization problem is iteratively solved until the objective function converges,
i.e., $\|\sum_{f\in\mathcal{F}}\sum_{n\in\mathcal{N}}\sum_{b\in\mathcal{B}}\sum_{k\in\mathcal{K}}s^{fn}_{bk}(\varrho)R^{fn}_{bk}(\textbf{p}(\varrho),\textbf{s}(\varrho_2),\textbf{g})-\\\sum_{f\in\mathcal{F}}\sum_{n\in\mathcal{N}}\sum_{b\in\mathcal{B}}\sum_{k\in\mathcal{K}}s^{fn}_{bk}(\varrho-1)
	R^{fn}_{bk}(\textbf{p}(\varrho),\textbf{s}(\varrho_2-1),\textbf{g})\|\leq {\epsilon}_2$. Note that Proposition 1 holds for small RBs assignment algorithm. The same approach used for D-RBS can be used to solve the optimization problem for S-RBS. Note that the optimization problem of the S-RBS scheme is similar to the D-RBS scheme, with the difference that in the S-RBS scheme optimization problem, there is no (\ref{eq-r21}).

\subsection{Convergence Analysis}
The convergence of joint power allocation and small RBs assignment algorithm, which is shown in Table \ref{Iterative Optimization Algorithm}, is investigated in the following Proposition.

\begin{prop}
	The joint power allocation and small RBs assignment algorithm, which is shown in Table \ref{Iterative Optimization Algorithm}, converges to a suboptimal solution of (\ref{eq-main}) which meets the KKT conditions of the optimization problem (\ref{eq-main}).
\end{prop}

\begin{proof}
It is shown in \cite[Subsection IV-A]{chiang2007power} and \cite{marks1978general} that the conditions for SCA convergence are guaranteed and the series of solutions obtained by the AGMA converges to a point where the KKT  conditions of Problem (\ref{eq-main}) are satisfied. 
For fixed RB assignment which is obtained in iteration $\varrho$ and the obtained value of power allocation in next iteration $\varrho+1$, the following inequality holds:
\begin{align}\label{eq-444}
\sum_{f\in\mathcal{F}}\sum_{n\in\mathcal{N}}\sum_{b\in\mathcal{B}}\sum_{k\in\mathcal{K}}s^{fn}_{bk}(\varrho)R^{fn}_{bk}(\textbf{p}(\varrho),\textbf{s}(\varrho),\textbf{g})\leq 
\sum_{f\in\mathcal{F}}\sum_{n\in\mathcal{N}}\sum_{b\in\mathcal{B}}\sum_{k\in\mathcal{K}}s^{fn}_{bk}(\varrho)R^{fn}_{bk}(\textbf{p}(\varrho+1),\textbf{s}(\varrho),\textbf{g}),
\end{align}
which stems from the fact that our goal is to maximize the objective function, and hence, the value of that must increases or remains fixed compared to its value in the previous iteration. Then, for given power allocation and obtaining RB assignment, we have 
\begin{align}\label{eq-555}
&\sum_{f\in\mathcal{F}}\sum_{n\in\mathcal{N}}\sum_{b\in\mathcal{B}}\sum_{k\in\mathcal{K}}s^{fn}_{bk}(\varrho)R^{fn}_{bk}(\textbf{p}(\varrho+1),\textbf{s}(\varrho),\textbf{g})\\\nonumber&\leq 
\sum_{f\in\mathcal{F}}\sum_{n\in\mathcal{N}}\sum_{b\in\mathcal{B}}\sum_{k\in\mathcal{K}}s^{fn}_{bk}(\varrho+1)R^{fn}_{bk}(\textbf{p}(\varrho+1),\textbf{s}(\varrho+1),\textbf{g}).
\end{align}
Similarly, the value of objective function must increases or remains fixed compared to its value in the previous iteration. Therefore, due to bounded feasibility set of the problem and the involved functions,  the algorithm converges to a sub-optimal solution.
\end{proof}

\section{Problem Solution by Monotonic Optimization}\label{monotonic}
In this section, we
adopt the monotonic optimization approach to find the optimal solution of the optimization problem (\ref{eq-main}). We show that joint power allocation and RB assignment problem can be reformulated
as a monotonic optimization problems.
The feasibility of optimization problem (\ref{eq-main}) is guaranteed, if $\textbf{p}$ satisfies each user's target minimum rate and (\ref{eqq-1}). The feasibility check procedure is shown in Table \ref{Projection Algorithm-44}. In this Table, we check the feasibility of $R^{\text{min,e}}_{bk}$ and $R^{\text{min,m}}_{bk}, \forall b,k$ when the transmit power $\textbf{p}$ is constrained by (\ref{eqq-1}). 
\begin{table}[h!]\caption{The Algorithm to Check the Feasibility of Optimization Problem (\ref{eq-main})}\label{Projection Algorithm-44}
	\centering
	\begin{tabular}{p{15cm}}
		\toprule
		\textbf{Algorithm II}: \\
		\midrule
		\textbf{Step1}: If the maximum
		eigenvalue of matrix $\boldsymbol{\Theta}$ is not smaller than 1, the minimum rate of each user is infeasible, else go to Step 2.\\
		\textbf{Step2}: The nonnegative power vector can be calculated as follows: $\textbf{p}=(\textbf{I}-\boldsymbol{\Theta})^{-1}\textbf{u}$, where $u^{fn}_{bk}=\frac{R^{\text{min,i}}_{bk}\vartheta N_0}{g^{fn}_{bk}}, \forall~ \text{if}~k\in \mathcal{K}^{\text{e}}_b~ \text{then}~\text{i}=\text{e}$, $\text{elseif}~k\in \mathcal{K}^{\text{m}}_b~ \text{then}~\text{i}=\text{m}$ and
		the elements of matrix $\boldsymbol{\Theta}$ are given by\\~~~~~~~~~~~~~~~~~~~~~~~~~~~~~~~~~~
		$\Theta^{fn}_{bk}=\begin{cases}
		0, & \text{if}~~ b = k.\\
		\frac{R^{\text{min,i}}_{bk}g^{fn}_{bk}}{g^{fn}_{bk}}, & \text{if}~~ b \neq k, \text{if}~k\in \mathcal{K}^{\text{e}}_b~ \text{then}~\text{i}=\text{e},~ \text{elseif}~k\in \mathcal{K}^{\text{m}}_b~ \text{then}~\text{i}=\text{m}.
		\end{cases}$\\~~~~~~~~
		If $\textbf{p}$ satisfies (\ref{eqq-1}), the minimum rate of each user is feasible.\\
		\bottomrule
		\label{Projection Algorithm-44}
	\end{tabular}
\end{table}
We define $\tilde{p}^{fn}_{bk}={s}^{fn}_{bk}{p}^{fn}_{bk}$ and rewrite optimization problem (\ref{eq-main}) in an equivalent form as follows:

\emph{\textbf{Problem}} $\mathcal{P}^{\text{D}}$
	\begin{align}\label{eq-main-22}
	\max_{\tilde{\textbf{p}},\textbf{s}}~~& R(\tilde{\textbf{p}}),
	\\\nonumber\label{eqq-1-mo}\text{s.t.}~~&
	(\ref{eq-min-R-e}), (\ref{eq-packet-loss})
	-(\ref{eq-packet-loss-2}), (\ref{eqq-1})-(\ref{eqq-3}).
	\end{align}
The optimization problem (\ref{eq-main-22}) is not convex. This problem can be rewritten as a canonical form of monotonic optimization as follows:

\emph{\textbf{Problem}} $\mathcal{P}^{\text{Monotonic}}$
\begin{align}
	\max_{\textbf{y}}~~&\sum_{f\in\mathcal{F}}\sum_{n\in\mathcal{N}}\sum_{b\in\mathcal{B}}\sum_{k\in\mathcal{
	K}}\log_2\left( 1+y^{fn}_{bk}\right) ,
	\\\label{eq-main-11-mono-111}\nonumber\text{s.t.}~~&
	\textbf{y}\in\mathcal{Y},
\end{align}
where $\textbf{y}=[y^{11}_{11},\dots,y^{fn}_{bk},\dots,y^{FN}_{BK}]$. The feasible set of the optimization problem is given by $\mathcal{Y}=\mathcal{G}\cap\mathcal{H}$ where
\begin{align}
\mathcal{G}=&\Big\lbrace \textbf{y} \big| y^{fn}_{bk} \leq  \gamma^{fn}_{bk}(\tilde{\textbf{p}}),
\tilde{\textbf{p}}\in\mathcal{P}, \textbf{s}\in\mathcal{S}
\Big\rbrace,
\end{align}
\begin{align}\nonumber
\mathcal{H}=&\Big\lbrace \textbf{y} \big|0\leq y^{fn}_{bk},\prod_{f\in\mathcal{F}}\prod_{n\in\mathcal{N}}(1+y^{fn}_{bk})\geq 2^{R^{\text{min,e}}_{bk}}, \forall  k\in\mathcal{K}^{\text{e}}_b,\prod_{f\in\mathcal{F}}\prod_{n\in\mathcal{N}}(1+y^{fn}_{bk})\geq 2^{R^{\text{min,m}}_{bk}}, \forall k\in\mathcal{K}^{\text{m}}_b,\\& \mathbb{E}_{\textbf{g}}\left\lbrace \Gamma(y^{fn}_{bk}) \right\rbrace  \leq \varepsilon^{\text{max},fn}_{bk}, \forall k\in\mathcal{K}^{\text{u}}_b,
\mathbb{E}_{\textbf{g}}\left\lbrace \Gamma(y^{fn}_{bk}) \right\rbrace  \leq \tilde{\varepsilon}^{\text{max},fn}_{bk}, \forall  k\in\mathcal{K}^{\text{m}}_b
\Big\rbrace,
\end{align}
where $\mathcal{P}$ and $\mathcal{S}$ are the feasible sets spanned by constraints
(\ref{eqq-1})-(\ref{eqq-3}) as follows:
\begin{equation}
\mathcal{P}=\Big\{\tilde{\textbf{p}} \big|\sum_{f\in\mathcal{F}}\sum_{n\in\mathcal{N}}\sum_{k\in\mathcal{K}_b}\tilde{p}^{fn}_{bk}\leq P^{\text{max}}_b, \tilde{p}^{fn}_{bk}\geq 0,\forall f\in\mathcal{F}, n\in\mathcal{N}, b\in\mathcal{B}, k\in\mathcal{K}\Big\rbrace,
\end{equation}
\begin{equation}
\mathcal{S}=\Big\{ \textbf{s} \big|  \sum_{k\in\mathcal{K}_b}s^{fn}_{bk}\leq 1, s^{fn}_{bk}\in\{0,1\},(\ref{eq-r21}), \forall f\in\mathcal{F}, n\in\mathcal{N}, b\in\mathcal{B}, k\in\mathcal{K}\Big\}.
\end{equation}
The optimal solution to optimization problem (\ref{eq-main-11-mono-111}) is denoted by $\textbf{y}^*=[(y^{11}_{11})^*,\dots,(y^{fn}_{bk})^*,\dots,\\(y^{FN}_{BK})^*]$. With the optimal value of  $\textbf{y}^*$ at hand, the optimal transmit powers, i.e., $\tilde{\textbf{p}}^*$, is the solution of $FNBK$ linear equations $(y^{fn}_{bk})^*(I^{fn}_{bk}+\vartheta N_0)-p^{fn}_{bk}g^{fn}_{bk}=0$ with $FNBK$ variables $p^{11}_{11},\dots,p^{fn}_{bk},\dots,p^{FN}_{BK}$.
Since the objective function $R(\textbf{y})$ is monotonic, we use monotonic optimization method to find its globally optimal solution \textcolor{blue}{for} our optimization problem. The key idea of global optimal algorithm is based on constructing a sequence of polyblock outer approximation of $\mathcal{Y}$, i.e., $\mathcal{R}_0\supset \mathcal{R}_1\supset\dots\supset \mathcal{Y}$, with its proper vertexes until the optimal
vertex of polyblock $\mathcal{R}_{\varrho}$ in the ${\varrho}^{\text{th}}$ iterative, lies in $\mathcal{Y}$. First, we construct an initial outer polyblock $\mathcal{R}_0$, which contains $\mathcal{Y}$, with one vertex $\textbf{y}_0=[y^{11}_{11,0},\dots,y^{fn}_{bk,0},\dots,y^{FN}_{BK,0}]$ and vertex set $\mathcal{T}_0=\{\textbf{v}_0\}$ where $\textbf{v}_0=\textbf{y}_0$. Then, $\textbf{y}_0$ is a global optimal solution and $\textbf{y}^*=\textbf{y}_0$, if $\textbf{y}_0\in\mathcal{Y}$; otherwise, construct a smaller outer polyblock $\mathcal{R}_1\subset \mathcal{R}_0$ of $\mathcal{Y}$ with vertex set $\mathcal{T}_1=\{\textbf{v}^{11}_{11,1},\dots,\textbf{v}^{fn}_{bk,1},\dots,\textbf{v}^{FN}_{BK,1}\}\}$,  where $\textbf{v}^{fn}_{bk,1}=\textbf{y}_0-(1-\beta_0)y^{fn}_{bk,0}\textbf{e}^{fn}_{bk}$ and $\beta_0$ is the projection of $\textbf{v}_0$ on the upper boundary of $\mathcal{G}$. Then, evaluate the objective function at each vertex in $\mathcal{T}_1$ to determine which one maximizes the objective function of (\ref{eq-main-11-mono-111}), i.e., $\textbf{y}_1=\arg\max_{\textbf{v}}\{R(\textbf{v})|\textbf{v}\in\mathcal{T}_1\}$. The procedure is repeated until $R(\textbf{y}_{\varrho_3})-R(\beta_{\varrho_3}\textbf{y}_{\varrho_3})\geq \delta$, where $\delta>0$. The elements of $\textbf{y}_0$ that satisfy the best SINR
for user $k$ from BS $b$ on RB $(f,n)$ can
be set to
\begin{equation}
y^{fn}_{bk,0}= P^{\text{max}}_bg^{fn}_{bk}/\vartheta N_0.
\end{equation}
Furthermore, the $\varrho_3^{\text{th}}$ projection, $\beta_{\varrho_3}$ can be calculated using $\beta_{\varrho_3}=\max\left\lbrace \beta|\beta \textbf{y}_{\varrho_3}\in \mathcal{Y} \right\rbrace $ which  is equivalent with the following problem:
\begin{subequations}\label{eq-BS}
\begin{align}
\max_{\beta,\tilde{\textbf{p}},\textbf{y}}~~&\beta,
\\\label{eqq-beta-0}\text{s.t.}~~&
\beta y^{fn}_{bk}(I^{fn}_{bk}+\vartheta N_0)\leq \tilde{p}^{fn}_{bk}g^{fn}_{bk},\forall f\in\mathcal{F}, n\in\mathcal{N}, b\in\mathcal{B},
\\&
\sum_{f\in\mathcal{F}}\sum_{n\in\mathcal{N}}\sum_{k\in\mathcal{K}_b}\tilde{p}^{fn}_{bk}\leq P^{\text{max}}_b\\
\label{eqq-beta-2}&
\prod_{f\in\mathcal{F}}\prod_{n\in\mathcal{N}}(1+\beta y^{fn}_{bk})\geq 2^{R^{\text{min,e}}_{bk}}, \forall b\in\mathcal{B}, k\in\mathcal{K}^{\text{e}}_b,\\\label{eqq-beta-3}&\prod_{f\in\mathcal{F}}\prod_{n\in\mathcal{N}}(1+\beta y^{fn}_{bk})\geq 2^{R^{\text{min,m}}_{bk}}, \forall b\in\mathcal{B}, k\in\mathcal{K}^{\text{m}}_b\\\label{eqq-beta-4}&
\mathbb{E}_{\textbf{g}}\left\lbrace \Gamma(\beta y^{fn}_{bk}) \right\rbrace  \leq \varepsilon^{\text{max},fn}_{bk}, \forall k\in\mathcal{K}^{\text{u}}_b\\\label{eqq-beta-5}&
\mathbb{E}_{\textbf{g}}\left\lbrace \Gamma(\beta y^{fn}_{bk}) \right\rbrace  \leq \tilde{\varepsilon}^{\text{max},fn}_{bk}, \forall  k\in\mathcal{K}^{\text{m}}_b.
\end{align}
\end{subequations}
The bisection method can be used to solve (\ref{eq-BS}), which is given in Table \ref{Projection Algorithm}.
\begin{table}[h!]\caption{The Bisection Algorithm for Calculating $\beta$}\label{Projection Algorithm}
	\centering
	\begin{tabular}{p{15cm}}
		\toprule
		\textbf{Algorithm III}: The Bisection Algorithm for Calculating $\beta$\\
		\midrule
		\textbf{Initialization} Set $Y_0=\max_{f,n,b,k}\{y^{fn}_{bk,0}\}$, $O_0=0$, $\varPi>0$ and $\varrho_3=0$;\\
		\textbf{While} $(Y_{\varrho_3}-O_{\varrho_3})/Y_{\varrho_3}\geq \varPi$\\~~~
		\textbf{Step1}: $\varrho_3=\varrho_3+1$;\\~~~
		\textbf{Step2}: 
		Update $\beta^*=(Y_{\varrho_3-1}+O_{\varrho_3-1})/2$;\\ ~~~
		\textbf{Step3}: Check whether $\tilde{p}^{fn}_{bk}$ satisfies constraints (\ref{eqq-beta-0})-(\ref{eqq-beta-5}) for a given $\beta^*$. \\~~~
		\textbf{Step4}: If finding $\tilde{p}^{fn}_{bk}$ is feasible (Check by Algorithm II), $Y_{\varrho_3}=Y_{\varrho_3-1}$, $O_{\varrho_3}=\beta^*$; otherwise $Y_{\varrho_3}=\beta^*$, $Y_{\varrho_3}=Y_{\varrho_3-1}$;\\
		\textbf{End} \\
		\textbf{Return}: Finally, 
		$\beta^*=\beta_{\varrho_3}$.\\
		\bottomrule
		\label{Projection Algorithm}
	\end{tabular}
\end{table}
Then, we can find the optimal solution of (\ref{eq-BS}), with $\textbf{y}^*=\beta^* \textbf{y}$, by solving
$FNBK$ linear equations $(y^{fn}_{bk})^*(I^{fn}_{bk}+\vartheta N_0)-p^{fn}_{bk}g^{fn}_{bk}=0$, which is
shown in Table \ref{Outer Polyblock Approximation Algorithm}. For any $\delta>0$, the
convergence analysis of Algorithm III is similar to Proposition 3.9 of \cite{zhang2013monotonic}. 
Finally, from optimal vertex $\textbf{y}^*$, we can obtained RB allocation $[s^{fn}_{bk}]^*$ as follows:
\begin{equation}
[s^{fn}_{bk}]^*=\begin{cases}
1, & \text{if}~~ [y^{fn}_{bk}]^*>0.\\
0, & \text{if}~~ [y^{fn}_{bk}]^*=0.
\end{cases}
\end{equation}
\begin{table}[h!]\caption{The Proposed Monotonic Algorithm for Optimization Problem}\label{Outer Polyblock Approximation Algorithm}
	\centering
	\begin{tabular}{p{15cm}}
		\toprule
		\textbf{Algorithm IV}: The Proposed Monotonic Algorithm for Optimization Problem\\
		\midrule
		\textbf{Feasibility Check}: Check the feasibility by Algorithm II.\\
		\textbf{Initialization} Construct an initial outer polyblock $\mathcal{R}^{(0)}$ of $\mathcal{Y}$ with one vertex  $\textbf{y}_0=[y^{11}_{11,0},\dots,y^{fn}_{bk,0},\dots,y^{FN}_{BK,0}]$ \\~~~~~~~~~~~~~~~~~where $y^{fn}_{bk,0}$ is given by  
		$y^{fn}_{bk,0}=P^{\text{max}}_bg^{fn}_{bk}/\vartheta N_0$; Set $\delta>0$ and iteration number $\varrho_4=1$.\\~~~
		\textbf{Step1}: Obtain $\beta$ using
		 Algorithm III.\\
		\textbf{While} $R(\textbf{y}_{\varrho_4})-R(\beta_{\varrho_4}\textbf{y}_{\varrho_4})\geq \delta$\\~~~
		\textbf{Step2}: $\varrho_4=\varrho_4+1$;\\~~~
		\textbf{Step3}: Generate a smaller polyblock $\mathcal{R}_{\varrho_4}$ with vertex set $\mathcal{T}_{\varrho_4}$ by replacing $\textbf{y}_{\varrho_4-1}$ with new vertices $\{\textbf{v}^{11}_{11,\varrho_4},\dots,\textbf{v}^{fn}_{bk,\varrho_4},\dots,\textbf{v}^{FN}_{BK,\varrho_4}\}$, where $\textbf{v}^{fn}_{bk,\varrho_4}=\textbf{y}_{\varrho_4-1}-(1-\beta_{\varrho_4-1})y^{fn}_{bk,\varrho_4-1}\textbf{e}^{fn}_{bk}$.\\~~~
		\textbf{Step4}: In set $\mathcal{T}_{(\varrho_4+1)}$, find vertex $\textbf{v}_{(\varrho_4+1)}$ which maximizes
		the objective function, i.e, $\textbf{y}_{(\varrho_4+1)}=\arg\max_{\textbf{v}\in\mathcal{T}_{(\varrho_4+1)}}\left\lbrace R(\textbf{v}) \right\rbrace$;\\~~~
		\textbf{Step5}: Calculate $\beta_{\varrho_4}$ using Algorithm III;\\
		\textbf{End} \\
		\textbf{Return}: Find optimal solution by solving  
		$FNBK$ linear equations $(y^{fn}_{bk})^*(I^{fn}_{bk}+\vartheta N_0)-p^{fn}_{bk}g^{fn}_{bk}=0$ with $\textbf{y}^*=\beta_{\varrho_4} \textbf{y}_{\varrho_4}$.\\
		\bottomrule
		\label{Outer Polyblock Approximation Algorithm}
	\end{tabular}
\end{table}

\section{Computational Complexity}\label{ComputationalComplexity andConvergence Analysis}
The computational complexity and the convergence of the resource allocation algorithms in both schemes are investigated in this section. We use CVX to solve the GP sub-problems
with the interior point method, therefore, the number of iterations needed to meet the required accuracy is $\frac{\log(\Gamma/\varrho_0\Psi)}{\log(\pi)}$, where $\Gamma$ is the total number of constraints, $\varrho_0$ is the initial point to approximate the accuracy of interior point method, $0<\Psi<<1$ is the stopping criterion for the interior point method and $\pi$ is used for updating the accuracy of interior point method \cite{boyd2004convex}. We assume $K^{\text{e}}_b=K^{\text{e}},\forall b$, $K^{\text{m}}_b=K^{\text{m}},\forall b$, and $K^{\text{u}}_b=K^{\text{u}},\forall b$. Then, for the D-RBS scenario, the number of constraints for power allocation and small RBs assignment sub-problems are $\Gamma_1=BK^{\text{e}}+BK^{\text{m}}+FNBK^{\text{u}}+FNBK^{\text{m}}+B+FNBK$, and $\Gamma_2=4 NBK^{\text{u}}+1+FNBK^{\text{u}}+BK^{\text{m}}+BK^{\text{e}}+FNBK^{\text{m}}+FNB+FNBK$, respectively. The
number of computations needed to convert power allocation and small RBs assignment non-convex sub-problems using AGMA are, respectively, $\Theta_1=3FNBK$, and $\Theta_2=FNB^2K+6FNBK+FNBK$. Hence, the computational complexity for power allocation and small RBs assignment sub-problems are given as follows:

\begin{align}\label{eq-com-1}
&\Theta_1\times\frac{\log(\Gamma_1/\varrho_0\Psi)}{\log(\pi)},~\Theta_2\times\frac{\log(\Gamma_2/\varrho_0\Psi)}{\log(\pi)},
\end{align}
respectively. For S-RBS, the total number of RBs is less than D-RBS, and (\ref{eq-r21}) is not considered in this scenario. Therefore, for S-RBS, the number of constraints for power allocation and RBs assignment sub-problems are $\tilde{\Gamma}_1=BK^{\text{e}}+BK^{\text{m}}+FNBK^{\text{u}}+FNBK^{\text{m}}+B+FNBK$, and $\tilde{\Gamma}_2=1+FNBK^{\text{u}}+BK^{\text{m}}+BK^{\text{e}}+FNBK^{\text{m}}+FNB+FNBK$, respectively. The
number of computations needed to convert power allocation and RBs assignment non-convex sub-problems using AGMA are, respectively, $\tilde{\Theta}_1=3FNBK$, and $\tilde{\Theta}_2=FNB^2K+6FNBK+FNBK$. Hence, the computational complexity for power allocation and RBs assignment sub-problems are given as follows:

\begin{align}\label{eq-com-2}
&\tilde{\Theta}_1\times\frac{\log(\tilde{\Gamma}_1/\varrho_0\Psi)}{\log(\pi)},~\tilde{\Theta}_2\times\frac{\log(\tilde{\Gamma}_2/\varrho_0\Psi)}{\log(\pi)}.
\end{align}

From Equations (\ref{eq-com-1}) and (\ref{eq-com-2}), it can be concluded that the computational complexity of D-RBS method is more than that of S-RBS method.

In the monotonic optimization approach, for a problem with dimensions $\hat{\tau}_1$, the number of iterations for obtaining the projection of each vertex by the bisection
algorithm is $\hat{\tau}_2$ and the number of iterations for the polyblock algorithm is $\hat{\tau}_3$, a simplified complexity order can be given by \cite{moltafet2018optimal}
\begin{equation}
\mathcal{O}(\hat{\tau}_3(\hat{\tau}_3\hat{\tau}_1+\hat{\tau}_2))
\end{equation}
The computational complexity of the optimal and suboptimal solutions are illustrated in Table \ref{11}. As can seen, with 5.3\% increasing in the computational complexity, we can achieve 8.6\% improvement in performance.
\begin{table*}[h!]
	\centering
	\caption{The computational complexity of the optimal and suboptimal solutions.}
	\label{11}
	\setlength{\tabcolsep}{3pt}
	\begin{tabular}{|c|c|c|}
		\hline
		& Optimal Solution over Suboptimal Solution & D-RBS over S-RBS, \\
		\hline
		\toprule
		Complexity & 5.3\% $\uparrow$ & 15.7\% $\uparrow$ 
		\\	\hline
		Performance & 8.6\% $\uparrow$ & 36\% $\uparrow$
		\\	\hline
	\end{tabular}
\end{table*}

\section{Simulation Results}\label{SimulationResults}
We consider a multi-cell downlink multi-user scenario with $B=4$ BSs and 240 kHz bandwidth serving $K^{\text{e}}_b=K^{\text{m}}_b=K^{\text{u}}_b=2$ users at each BS which are uniformly distributed in a 500m$\times$500m square area. The channel gains are modelled as $g^{fn}_{bk}=\varkappa^{fn}_{bk}L^{-c}_{bk}$ where $c=3$ is the path loss exponent, $L_{bk}$ is the normalized distance
between the BS $b$ and user $k$ and $\varkappa^{fn}_{bk}$ is exponentially distributed with mean 1 \cite{she2018joint}. The noise power of each sub-carrier is normalized to
1 or 0 dB. We also set $\epsilon_1=10^{-5}$, $\epsilon_2=10^{-5}$ and $\xi_2=10^5$ in all simulations \cite{3GPP,3GPP-36-211}. For 1000 random channel realizations, our algorithm is run and the average of results are presented. In Fig. \ref{Frame_Structure_Two_Scenario}, the structure of RBs of S-RBS and D-RBS scenarios are illustrated for 240 kHz bandwidth and 1 ms TTI. For D-RBS, we assume RBs of equal sizes in both time and frequency, i.e., TTI=0.25 ms and BW=15 kHz which results in 16$\times$4 RBs. In S-RBS, we assume static frame structure for each cell. At each cell, there are 16 RBs: four RBs with TTI=1 ms and BW=15 kHz, four RBs with TTI=0.5 ms and BW=30 kHz and eight RBs with TTI=0.25 ms and BW=60 kHz.
We assume that the time duration and bandwidth of each RB are shorter than the channel coherence time and bandwidth, respectively. 


\begin{figure}[h]
    \begin{center}
        \includegraphics[width=4 in]{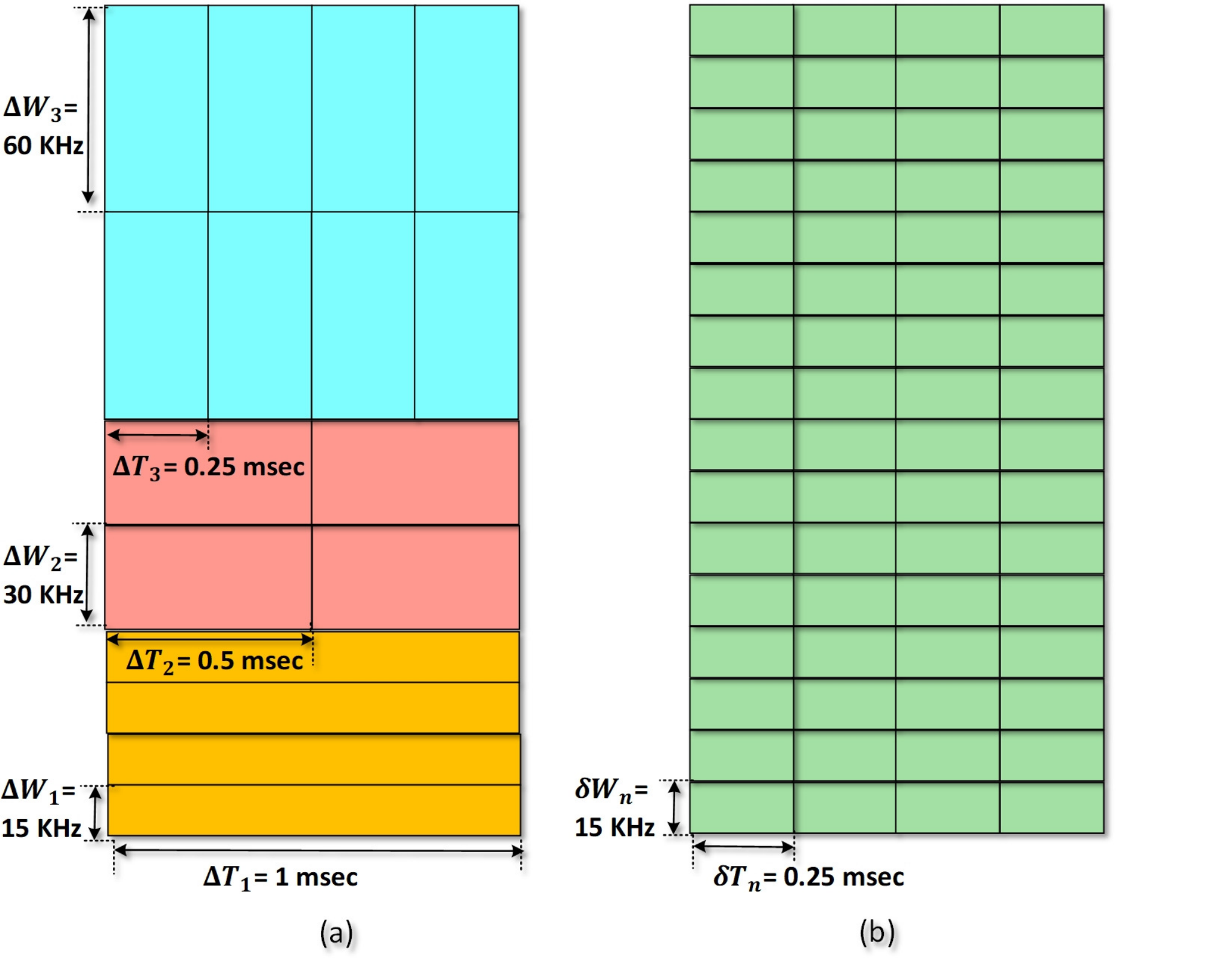} 
        \caption{The structure of RBs of: a) S-RBS scenario and b) D-RBS scenario}
        \label{Frame_Structure_Two_Scenario}
    \end{center}
\end{figure}

\subsection{Effect of total number of RBs}
The sum rate versus the total number of RBs for both scenarios is illustrated in Fig. \ref{Effect of RBs}. As can be seen, due to the opportunistic nature of fading channels, the sum rate in both scenarios increases as the total number of RBs increases. In the second scenario, due to flexibility of scheduling structure, improvement  in terms of sum rate is better than the first scenario. The figure also shows that the sum rate is decreased by increasing the value of $R^{\text{min,e}}$, and $R^{\text{min,m}}$. This decrease is due to the fact that with increasing the value of $R^{\text{min,e}}$, and $R^{\text{min,m}}$, the feasibility region of RBs assignment and power allocation become smaller and leading to less sum rate in both scenarios. However, It should be noted that this decline in the second scenario is less due to the flexible allocation of resources compared to the first scenario. In both scenarios, we set ${\varepsilon}^{\text{max},fn}_{bk}={\varepsilon}^{\text{max}}=10^{-5}, \forall f\in\mathcal{F}, n\in\mathcal{N}, b\in\mathcal{B}, k\in\mathcal{K}^{\text{u}}_b$, $\tilde{\varepsilon}^{\text{max},fn}_{bk}=\tilde{\varepsilon}^{\text{max}}=10^{-3}, \forall f\in\mathcal{F}, n\in\mathcal{N}, b\in\mathcal{B}, k\in\mathcal{K}^{\text{e}}_b$, $R^{\text{min,e}}_{bk}=R^{\text{min,e}}, \forall b\in\mathcal{B}, k\in\mathcal{K}^{\text{e}}_b$, $R^{\text{min,m}}_{bk}=R^{\text{min,m}}, \forall b\in\mathcal{B}, k\in\mathcal{K}^{\text{m}}_b$, and $P^{\text{max}}_b=P^{\text{max}}=$ 40 dB, $\forall b\in\mathcal{B}$ \cite{popovski20185g}. We also compare the performance of optimal and suboptimal solutions in Fig. \ref{Effect of RBs}. As can be seen, our suboptimal solution is bounded by optimal solution and it  has values close to optimal one.
%

\begin{figure}[h]
	\begin{center}
		\subfigure[]{\label{Effect of RBs}
			\includegraphics[width=4 in]
			{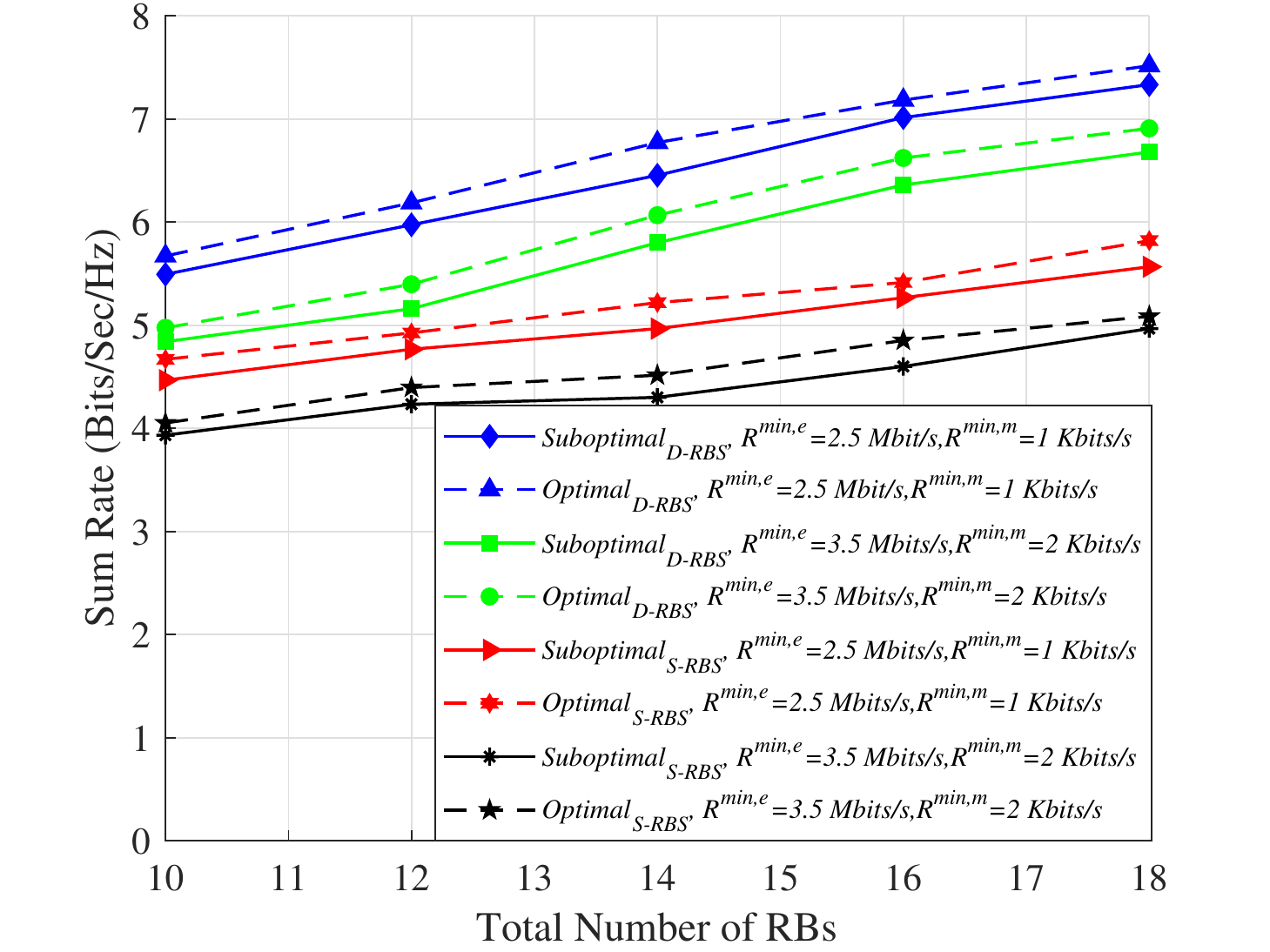}}
		\subfigure[]{\label{Fig_K_e_K_u}
			\includegraphics[width=3.6 in]
			{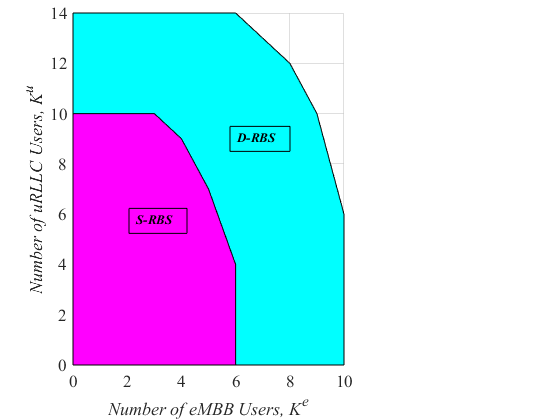}}
		\caption{(a) The sum rate versus the total number of BRs for both scenarios for different values of $R^{\text{min,e}}$, and $R^{\text{min,m}}$ for optimal and suboptimal solutions. (b) The regions of D-RBS and S-RBS that satisfy URLLC and eMBB users' requirements.}
	\end{center}
\end{figure}

\subsection{Effect of Number of Users}
In Fig. \ref{Fig_K_e_K_u},  we present simulation results for the regions of D-RBS and S-RBS that satisfy URLLC and eMBB users' requirements versus the number of URLLC and eMBB users. In each scheme, the RB is chosen such that the target latency, reliability and throughput of URLLC and eMBB users are satisfied. By increasing the number of users, we see that the D-RBS scheme outperforms the S-RBS scheme due to the flexible RB adaptation. In each scheme, to obtain feasible regions, for fixed number of eMBB or URLLC users, we obtain the number of other users can be serviced by same number of RBs.
%

\subsection{The Outage Probability}
In the following, we define outage probabilities as following to investigate behavior of both scenarios for different values of $R^{\text{min,e}}_{bk}$, and $R^{\text{min,m}}_{bk}$
\begin{align}
&\mathcal{P}^{\text{e}}=\Pr\left( \sum_{f\in\mathcal{F}}\sum_{n\in\mathcal{N}}s^{fn}_{bk}R^{fn}_{bk}(\textbf{p},\textbf{s},\textbf{p})\leq R^{\text{min,e}}_{bk}\right) , \forall b\in\mathcal{B}, k\in\mathcal{K}^{\text{e}}_b,\\&\mathcal{P}^{\text{m}}=\Pr\left(\sum_{f\in\mathcal{F}}\sum_{n\in\mathcal{N}}s^{fn}_{bk}R^{fn}_{bk}(\textbf{p},\textbf{s},\textbf{p})\leq R^{\text{min,m}}_{bk}\right) , \forall b\in\mathcal{B}, k\in\mathcal{K}^{\text{m}}_b.
\end{align}

We obtain the outage probability of both scenarios via Mont Carlo simulation. The outage
probability of mMTC users for both scenarios versus different values of $R^{\text{min,e}}$ is shown in Fig. \ref{Fig_Rate_outage_vs_R_min_e_diff_scenarios}. As can be seen, with increasing $R^{\text{min,e}}$, the outage probability is also increased for both scenarios. However, due to the larger feasibility region, the second scenario has lower outage probability
compared to the first one. On the other hand, the second scenario can efficiently schedule RBs between different users as compared to the first scenario. In other words, flexible scheduling structure in the second scenario has more degrees of freedom
to assign suitable RBs users
of different cells. Therefore, due to higher outage probability in the first scenario, it cannot satisfy the minimum rate requirements of users.

\begin{figure}[h]
	\begin{center}
		\subfigure[]{\label{Fig_Rate_outage_vs_R_min_e_diff_scenarios}
			\includegraphics[width=3.1 in]
			{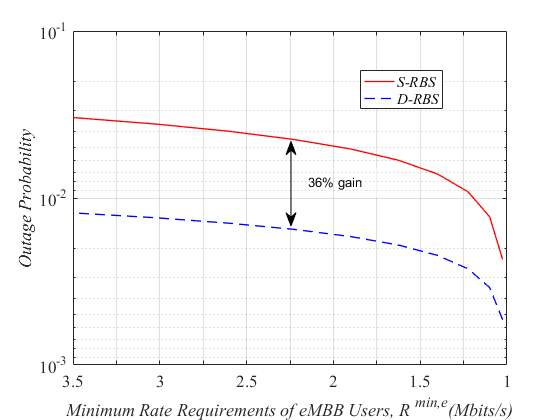}}
		\subfigure[]{\label{Fig_Rate_outage_vs_R_min_m_diff_scenarios}
			\includegraphics[width=3.1 in]
			{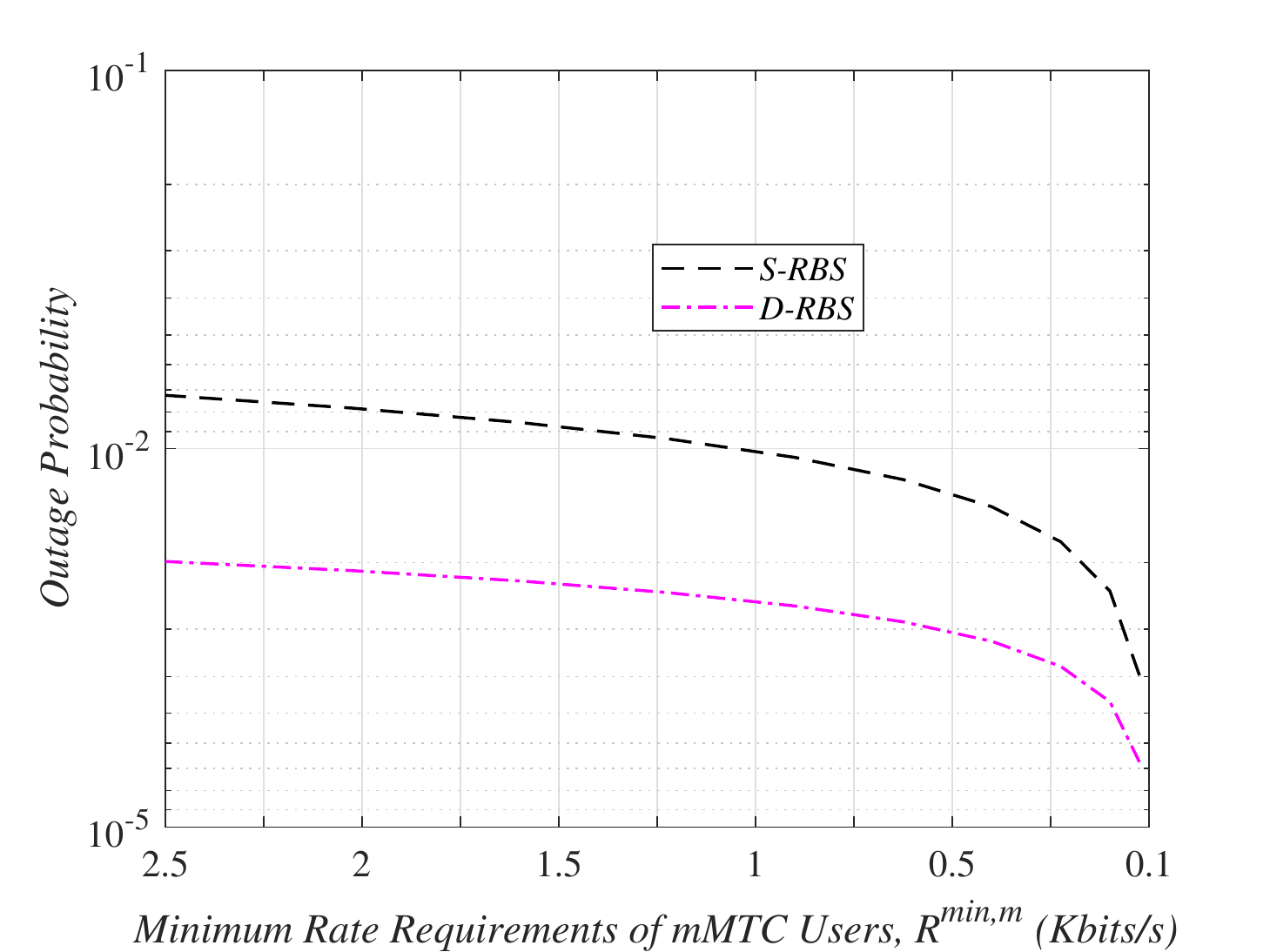}}
		\caption{(a) The outage probability versus the different values of $R^{\text{min,e}}$ for both scenarios. (b) The outage probability versus the different values of $R^{\text{min,m}}$ for both scenarios.}
	\end{center}
\end{figure}


The outage probability versus the value of minimum rate requirements of mMTC users is illustrated in Fig. \ref{Fig_Rate_outage_vs_R_min_m_diff_scenarios}.

%

\subsection{Convergence Analysis}
In Fig. \ref{Fig_Number_of_Iterations_vs_RBs_diff_scenarios}, the total number of iterations needed for the algorithm to converge for both scenarios versus the total number of RBs is depicted. It is seen from the figure that as  the number of RBs increases, the algorithm needs more iteration to converge. It can also be seen that the computational complexity of RBs assignment problem for both scenarios is higher than
that of power allocation problem because the number of constraints in RBs assignment problem is higher than power allocation problem. Also, the computational complexity of D-RBS is higher than S-RBS because the number of constraints in D-RBS is higher than S-RBS.
\begin{figure}[h]
    \begin{center}
        \includegraphics[width=4.5 in]{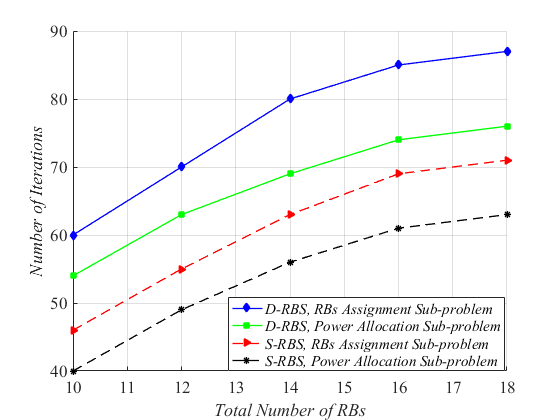} 
        \caption{The number of iterations required for convergence for both power allocation and RBs assignment sub-problems for both scenarios versus the total number of RBs.} 
        \label{Fig_Number_of_Iterations_vs_RBs_diff_scenarios}
    \end{center}
\end{figure}

\section{Conclusion}\label{conclusions}
In this paper, we developed a flexible and dynamic BRs assignment and transmit power allocation for  multi-user multi-cell downlink scheme to satisfy every user's rate and PER requirements. Our proposed scheduling frame work, due to the large feasibility region and more degrees of freedom to assign RBs had better performance than the traditional RBs scheme. To solve the resulting non-convex optimization problem, we relied in ASM and exploited the successive convex approximation method to obtain a convex approximation of the original problem which could be solved by existing tools like CVX. Via simulations, we showed that the proposed scheme has 26\% and 36\% performance gain, in the sum rate and outage probability, respectively, over the S-RBS scheme.

\hyphenation{op-tical net-works semi-conduc-tor}
\bibliographystyle{IEEEtran}
\bibliography{IEEEabrv,Bibliography}
\end{document}